\documentclass[11pt,notitlepage,tightenlines,nofootinbib,superscriptaddress]{revtex4-2}

\usepackage{algorithm}
\usepackage{algpseudocode}

\floatname{algorithm}{Protocol}
\usepackage[dvipsnames]{xcolor}
\usepackage{color}

\let\coloneqq\relax

\usepackage[latin1]{inputenc}
\usepackage{amsthm}
\usepackage{amssymb}
\usepackage{amsmath}
\usepackage{bbold}
\usepackage{bbm}
\usepackage[pdftex, backref=page]{hyperref}
\usepackage{dsfont}
\usepackage{mathdots}
\usepackage{mathtools}
\usepackage{enumerate}
\usepackage[shortlabels]{enumitem}
\usepackage{csquotes}
\usepackage{stmaryrd}
\usepackage[cal=boondox]{mathalfa}
\usepackage{graphicx}
\usepackage{stackengine}
\usepackage{scalerel}
\usepackage{tensor}       
\usepackage{array}
\usepackage{makecell}
\newcolumntype{x}[1]{>{\centering\arraybackslash}p{#1}}
\usepackage{tikz}
\usepackage{pgfplots}
\usetikzlibrary{shapes.geometric, shapes.misc, positioning, arrows, arrows.meta, decorations.pathreplacing, decorations.pathmorphing, patterns, angles, quotes, calc}
\usepackage{booktabs}
\usepackage{xfrac}
\usepackage{siunitx}
\usepackage{centernot}
\usepackage{comment}
\usepackage{chngcntr}
\usepackage{caption}
\usepackage{subcaption}
\usepackage{physics}
\usepackage{epigraph}

\newtheorem{thm}{Theorem}
\newtheorem*{thm*}{Theorem}
\newtheorem{prop}[thm]{Proposition}
\newtheorem*{prop*}{Proposition}
\newtheorem{lemma}[thm]{Lemma}
\newtheorem*{lemma*}{Lemma}
\newtheorem{cor}[thm]{Corollary}
\newtheorem*{cor*}{Corollary}

\newtheorem*{cj*}{Conjecture}
\newtheorem{Def}[thm]{Definition}
\newtheorem*{Def*}{Definition}

\newtheorem*{question*}{Question}

\newtheorem*{problem*}{Problem}

\makeatletter
\def\thmhead@plain#1#2#3{%
  \thmname{#1}\thmnumber{\@ifnotempty{#1}{ }\@upn{#2}}%
  \thmnote{ {\the\thm@notefont#3}}}
\let\thmhead\thmhead@plain
\makeatother

\theoremstyle{definition}
\newtheorem{rem}[thm]{Remark}
\newtheorem*{note}{Note}

\newtheorem{manuallemmainner}{Lemma}
\newenvironment{manuallemma}[1]{%
  \renewcommand\themanuallemmainner{#1}%
  \manuallemmainner \it
}{\endmanuallemmainner}

\DeclarePairedDelimiter\ceil{\lceil}{\rceil}
\DeclarePairedDelimiter\floor{\lfloor}{\rfloor}

\newcommand{\bb}{\begin{equation}\begin{aligned}\hspace{0pt}}
\newcommand{\bbb}{\begin{equation*}\begin{aligned}}
\newcommand{\ee}{\end{aligned}\end{equation}}
\newcommand{\eee}{\end{aligned}\end{equation*}}
\newcommand*{\coloneqq}{\mathrel{\vcenter{\baselineskip0.5ex \lineskiplimit0pt \hbox{\scriptsize.}\hbox{\scriptsize.}}} =}

\newcommand{\eqt}[1]{\stackrel{\mathclap{\scriptsize \mbox{#1}}}{=}}
\newcommand{\leqt}[1]{\stackrel{\mathclap{\scriptsize \mbox{#1}}}{\leq}}

\newcommand{\geqt}[1]{\stackrel{\mathclap{\scriptsize \mbox{#1}}}{\geq}}
\renewcommand{\ketbra}[1]{\ket{#1}\!\!\bra{#1}}

\newcommand{\ketbraa}[2]{\ket{#1}\!\!\bra{#2}}

\newcommand{\sumno}{\sum\nolimits}

\newcommand{\e}{\varepsilon}
\renewcommand{\epsilon}{\varepsilon}

\newcommand{\id}{\mathds{1}}

\newcommand{\N}{\mathds{N}}

\newcommand{\C}{\mathds{C}}

\newcommand{\locc}{\mathrm{LOCC}}
\newcommand{\sep}{\mathrm{SEP}}

\newcommand{\SEP}{\pazocal{S}}

\DeclareMathOperator{\rk}{rk}
\DeclareMathOperator{\cl}{cl}
\DeclareMathOperator{\co}{conv}

\DeclareMathAlphabet{\pazocal}{OMS}{zplm}{m}{n}

\DeclareMathOperator{\pr}{Pr}
\DeclareMathOperator{\supp}{supp}
\DeclareMathOperator{\spec}{spec}

\DeclareMathOperator{\dom}{dom}

\newcommand{\HH}{\pazocal{H}}

\newcommand{\EE}{\pazocal{E}}

\newcommand{\XX}{\pazocal{X}}

\newcommand{\lsmatrix}{\left(\begin{smallmatrix}}
\newcommand{\rsmatrix}{\end{smallmatrix}\right)}

\newcommand{\deff}[1]{\textbf{\emph{#1}}}

\stackMath
\newcommand\xxrightarrow[2][]{\mathrel{%
  \setbox2=\hbox{\stackon{\scriptstyle#1}{\scriptstyle#2}}%
  \stackunder[5pt]{%
    \xrightarrow{\makebox[\dimexpr\wd2\relax]{$\scriptstyle#2$}}%
  }{%
   \scriptstyle#1\,%
  }%
}}

\newcommand{\tends}[2]{\xxrightarrow[\! #2 \!]{\mathrm{#1}}}
\newcommand{\tendsn}[1]{\xxrightarrow[\! n\rightarrow \infty\!]{\mathrm{#1}}}

\stackMath

\makeatletter
\newcommand*\rel@kern[1]{\kern#1\dimexpr\macc@kerna}
\newcommand*\widebar[1]{%
  \begingroup
  \def\mathaccent##1##2{%
    \rel@kern{0.8}%
    \overline{\rel@kern{-0.8}\macc@nucleus\rel@kern{0.2}}%
    \rel@kern{-0.2}%
  }%
  \macc@depth\@ne
  \let\math@bgroup\@empty \let\math@egroup\macc@set@skewchar
  \mathsurround\z@ \frozen@everymath{\mathgroup\macc@group\relax}%
  \macc@set@skewchar\relax
  \let\mathaccentV\macc@nested@a
  \macc@nested@a\relax111{#1}%
  \endgroup
}

\counterwithin*{equation}{part}
\counterwithin*{thm}{part}
\counterwithin*{figure}{part}

\tikzset{meter/.append style={draw, inner sep=10, rectangle, font=\vphantom{A}, minimum width=30, line width=.8, path picture={\draw[black] ([shift={(.1,.3)}]path picture bounding box.south west) to[bend left=50] ([shift={(-.1,.3)}]path picture bounding box.south east);\draw[black,-latex] ([shift={(0,.1)}]path picture bounding box.south) -- ([shift={(.3,-.1)}]path picture bounding box.north);}}}
\tikzset{roundnode/.append style={circle, draw=black, fill=gray!20, thick, minimum size=10mm}}
\tikzset{squarenode/.style={rectangle, draw=black, fill=none, thick, minimum size=10mm}}

\definecolor{Blues5seq1}{RGB}{239,243,255}
\definecolor{Blues5seq2}{RGB}{189,215,231}
\definecolor{Blues5seq3}{RGB}{107,174,214}
\definecolor{Blues5seq4}{RGB}{49,130,189}
\definecolor{Blues5seq5}{RGB}{8,81,156}

\definecolor{Greens5seq1}{RGB}{237,248,233}
\definecolor{Greens5seq2}{RGB}{186,228,179}
\definecolor{Greens5seq3}{RGB}{116,196,118}
\definecolor{Greens5seq4}{RGB}{49,163,84}
\definecolor{Greens5seq5}{RGB}{0,109,44}

\definecolor{Reds5seq1}{RGB}{254,229,217}
\definecolor{Reds5seq2}{RGB}{252,174,145}
\definecolor{Reds5seq3}{RGB}{251,106,74}
\definecolor{Reds5seq4}{RGB}{222,45,38}
\definecolor{Reds5seq5}{RGB}{165,15,21}

\allowdisplaybreaks

\hypersetup{
    bookmarksnumbered=true, 
    unicode=false, 
    pdfstartview={FitH}, 
    pdftitle={}, 
    pdfauthor={}, 
    pdfsubject={}, 
    pdfcreator={}, 
    pdfproducer={}, 
    pdfkeywords={}, 
    pdfnewwindow=true, 
    colorlinks=true, 
    linkcolor=NavyBlue, 
    citecolor=NavyBlue, 
    filecolor=NavyBlue, 
    urlcolor=NavyBlue 
}

\renewcommand{\deff}[1]{\textbf{\textit{#1}}}
\renewcommand{\HH}{\mathcal{H}}

\begin{document}

\title{Entanglement cost for infinite-dimensional physical systems}

\author{Hayata Yamasaki}
\email{hayata.yamasaki@gmail.com}
\affiliation{Department of Physics, Graduate School of Science, The Univerisity of Tokyo, 7--3--1 Hongo, Bunkyo-ku, Tokyo, 113--0033, Japan}
\affiliation{Department of Computer Science, Graduate School of Information Science and Technology, The University of Tokyo, 7--3--1 Hongo, Bunkyo-ku, Tokyo, 113--8656, Japan}

\author{Kohdai Kuroiwa}
\email{kkuroiwa@uwaterloo.ca}
\affiliation{Institute for Quantum Computing and Department of Combinatorics and Optimization, University of Waterloo, Ontario, Canada, N2L 3G1}
\affiliation{Perimeter Institute for Theoretical Physics, Ontario, Canada, N2L 2Y5}

\author{Patrick Hayden}
\email{phayden@stanford.edu}
\affiliation{Stanford Institute for Theoretical Physics, Stanford University, Stanford, California 94305, USA}

\author{Ludovico Lami}
\email{ludovico.lami@gmail.com}
\affiliation{Scuola Normale Superiore, Piazza dei Cavalieri 7, 56126 Pisa, Italy}
\affiliation{QuSoft, Science Park 123, 1098 XG Amsterdam, The Netherlands}
\affiliation{Korteweg--de Vries Institute for Mathematics, University of Amsterdam, Science Park 105--107, 1098 XG Amsterdam, The Netherlands}
\affiliation{Institute for Theoretical Physics, University of Amsterdam, Science Park 904, 1098 XH Amsterdam, The Netherlands}

\begin{abstract}
We prove that the entanglement cost equals the regularized entanglement of formation for any infinite-dimensional quantum state $\rho_{AB}$ with finite quantum entropy on at least one of the subsystems $A$ or $B$. This generalizes a foundational result in quantum information theory that was previously formulated only for operations and states on finite-dimensional systems. The extension to infinite-dimensional systems is nontrivial because the conventional tools for establishing both the direct and converse bounds, i.e., strong typicality, monotonicity, and asymptotic continuity, are no longer directly applicable. To address this problem, we construct a new entanglement dilution protocol for infinite-dimensional states implementable by local operations and a finite amount of one-way classical communication (one-way LOCC), using weak and strong typicality multiple times. We also prove the optimality of this protocol among all protocols, even under infinite-dimensional separable operations, by developing an argument based on alternative forms of monotonicity and asymptotic continuity of the entanglement of formation for infinite-dimensional states. Along the way, we derive a new integral representation for the quantum entropy of infinite-dimensional states, which we believe to be of independent interest. Our results allow us to fully characterize an important operational entanglement measure --- the entanglement cost --- for all infinite-dimensional physical systems.
\end{abstract}

\maketitle

\tableofcontents
\newpage

\section{Introduction: the asymptotic continuity catastrophe}

\epigraph{
\emph{Midway upon the journey of our gates,} \\
\emph{we found ourselves within a system large,} \\
\emph{for the dimension finite had been lost.}
}{Quantum Inferno, Canto~1}

\subsection{Entanglement dilution}

A fundamental discovery of quantum information science is that quantum entanglement is not only one of the most striking features of quantum mechanics~\cite{schr}, but it can be considered from the operational side as a \textit{resource}, namely, as a fuel that powers quantum technology. Among its many applications, we find quantum teleportation~\cite{teleportation}, dense coding~\cite{dense-coding}, the violation of Bell inequalities~\cite{Brunner-review}, and quantum key distribution~\cite{Ekert91, RennerPhD}.
As any other resource, quantum entanglement can be manipulated, i.e., transformed into different forms. The canonical way to do so is by employing \textit{local operations and classical communication} (LOCC), which constitute the most general operations that distant parties bound to the rule of quantum mechanics and allowed to exchange only classical messages can implement~\cite{LOCC}.

The most important entanglement manipulation protocol is perhaps \textit{entanglement distillation}~\cite{Bennett-distillation, Bennett-distillation-mixed, Bennett-error-correction}, whose goal is to transform many identical and ideally distributed (IID) copies of a quantum state $\rho_{AB}$ into as many copies as possible of the \textit{ebit}, i.e., the two-qubit maximally entangled state shared between $A$ and $B$ while making an asymptotically vanishing error. The converse task is, however, equally important from the conceptual~\cite{Bennett-error-correction} as well as possibly practical~\cite{Miller2022} point of view: it consists of using LOCC assisted by pure entanglement, in the form of ebits, to prepare as many copies of $\rho_{AB}$ as possible, again with asymptotically vanishing error. It is therefore called \textit{entanglement dilution}. The duality between entanglement distillation and entanglement dilution offers a striking parallel with classical thermodynamics, where heat can be transferred from hot to cold baths to obtain work, or work can be invested to transfer heat in the opposite direction. Entanglement is, however, more complex than what this intriguing similarity might suggest. In fact, while in classical thermodynamics work and heat can be interconverted reversibly by means of Carnot cycles, entanglement theory harbors the fascinating phenomenon of irreversibility, which shows that different forms of entanglement are in general inequivalent~\cite{Vidal-irreversibility, faithful-EC, irreversibility-PPT, irreversibility, irreversibility-channels, probabilistic-reversibility}.

One of the pillars of finite-dimensional entanglement theory is the characterization of the optimal rate of ebits
that must be consumed per output copy of $\rho = \rho_{AB}$ in entanglement dilution under LOCC\@. This rate is called the \textit{entanglement cost} of $\rho_{AB}$, denoted by $E_c\qty(\rho)$. 
As proved in Ref.~\cite{Hayden-EC} in the case of finite-dimensional quantum systems (see also~\cite{Buscemi2011}), it is given by 
\bb
E_c\qty(\rho) = E_f^\infty\qty(\rho) \coloneqq \lim_{n\to\infty} \frac1n\, E_f\big(\rho^{\otimes n}\big)\, ,
\label{HHT}
\ee
where $E_f$, called the \textit{entanglement of formation}~\cite{Bennett-error-correction}, is defined as the convex-roof extension of the entanglement entropy~\cite{Bennett-distillation}, and $E_f^\infty$ is its regularized form.

\subsection{Enter infinite-dimensional systems} 

So far so good. But, problematically, the quantum physics in our world is not necessarily as straightforward as the analysis in the finite-dimensional case might suggest. To see why one needs to consider an infinite-dimensional system, consider that even within the paradigm of current quantum technology, we can easily step out of the reassuring finite-dimensional setting. 
In actual quantum experiments, a ``perfect'' qubit, i.e., an intrinsically two-dimensional quantum system without any additional degrees of freedom, is hardly ever encountered.
Rather, we make qubits by embedding them into an infinite-dimensional system in various ways.
This is what happens, e.g., for superconducting qubits, trapped ions, and single photons.
In these experimental platforms, a subspace of the overall system spanned by the ground state and the first excited state can be used as a qubit; at the same time, the overall Hilbert space is intrinsically infinite-dimensional.

Regardless of whether quantum information processing is designed to be qubit-based, infinite-dimensional Hilbert space is bound to lurk just above the finite-dimensional subspace where we carry out our computations.
Indeed, the effects of state transitions leaking into higher energy levels are common in practical implementations of quantum protocols. 
With more precise control, in principle, we can potentially take advantage of the higher dimensions of Hilbert space to enhance the performance of quantum information processing.

Also from the practical standpoint, one of the most important applications of entanglement theory is the theory of entanglement distribution in quantum optical networks, which are intrinsically infinite-dimensional objects~\cite{Pirandola2020, Pirandola2021}. The many results involving the calculation or the estimation of capacities of continuous-variable channels, such as the quantum capacities of the pure loss channel~\cite{holwer, Wolf2007, Pirandola2009, Mark2012, TGW, PLOB, MMMM}, of the thermal attenuator~\cite{Pirandola2009, PLOB, Rosati2018, Noh2020, FKG, lower-bound, KFG}, and of non-Gaussian bosonic dephasing channels~\cite{exact-solution}, as well as the calculation of the classical capacity of phase-insensitive single-mode Gaussian channels~\cite{Giovadd, Giovadd-CMP}, are exemplary of this widespread interest. Recently, also the entanglement cost of continuous-variable quantum channels has been thoroughly investigated~\cite{Wilde2018}.

Finally, from a more fundamental standpoint, quantum field theory suggests that the elementary constituents of matter are represented by infinite-dimensional quantum systems. Accordingly, entanglement of quantum fields has become the subject of increasing interest~\cite{HOLLANDS}. 

Whether the motivation lies in practical applications or in the fundamental logical consistency of the theory, the question is the same: \textit{what are the limitations on entanglement manipulation in infinite-dimensional systems?} To answer this question, in recent years, there has been a growing interest in the generalization of results in entanglement theory from finite- to infinite-dimensional quantum systems. 

\subsection{Main contribution} \label{subsec_main_contribution}

The present paper aims to generalize the formula~\eqref{HHT} for the entanglement cost to all infinite-dimensional quantum systems with an underlying separable Hilbert space.

Several issues arise when one attempts to extend the original proof, which works only in finite dimension, to the infinite-dimensional case.
Let us leave aside for a second the fact that even the definition of entanglement of formation presents some subtleties in the latter case~\cite{Shirokov2010}, as those can be avoided by restricting to sufficiently well-behaved, e.g., finite-energy, quantum states.
The proof of~\eqref{HHT} in Ref.~\cite{Hayden-EC} is composed of two main parts: the direct part, which consists of constructing an entanglement dilution protocol that achieves an asymptotic rate close to the right-hand side of~\eqref{HHT}, and the converse part, whose goal is to show that this is the optimal rate, i.e., that no other protocol can achieve a lower rate.

Neither the direct nor the converse part of the argument in Ref.~\cite{Hayden-EC} work as stated in infinite-dimensional systems. Regarding the direct part, the key technique in constructing the entanglement dilution protocol in~\cite{Hayden-EC} was strong typicality, which means that if a sequence $\qty(x_1,\ldots,x_N)\in\mathcal{X}^N$ is sampled in an independent and identically distributed (IID) way from a probability distribution $p\qty(x)$ over a finite set $\mathcal{X}$, then each element $x\in\mathcal{X}$ appears in this sequence approximately $N p\qty(x)$ times.
The problem here lies in the fact that strong typicality holds only for finite sets $\mathcal{X}$.
It is necessary to take into account infinite sets if one wants to generalize the result to infinite-dimensional systems.

An equally serious issue arises in the converse part of the proof. Remarkably, the issue arises even if the target state $\rho_{AB}$ itself is supported on the tensor product of finite-dimensional local spaces. 
Indeed, it is in principle possible that some LOCC protocols could aim to obtain an approximation to several copies of $\rho_{AB}$, i.e., to the state $\rho_{AB}^{\otimes n}$, by actually preparing an infinite-dimensional state that acts on the full Hilbert space but is nonetheless close to $\rho_{AB}^{\otimes n}$ in trace distance. Such protocols could have, in principle, a lower entanglement cost than any finite-dimensional protocol. 
The key technical tool that is employed in the finite-dimensional proof to control the approximation to the target state is the \textit{asymptotic continuity} of the entanglement of formation $E_f$~\cite{Nielsen2000, tightuniform,Mark2020}.
This is a particularly strong form of continuity in which the dimension of the underlying space makes an explicit appearance. For infinite-dimensional spaces, however, the statement trivializes.
Due to this phenomenon, dubbed the ``asymptotic continuity catastrophe'' in Ref.~\cite{nonclassicality}, it was impossible, prior to our work, to establish a converse bound for entanglement dilution in terms of the regularized entanglement of formation. Some converse bounds in infinite-dimensional systems were available~\cite{nonclassicality, Kuroiwa2021}, but they did not coincide with the regularized entanglement of formation, and thus were not sufficient to establish~\eqref{HHT}.

A somewhat different case is that of Ref.~\cite[Proposition~5]{Wilde2018}, which proves that the entanglement cost of certain bosonic Gaussian channels is lower bounded by their regularized entanglement of formation. This is an example of a converse statement for entanglement dilution obtained in an infinite-dimensional setting. So, why is this proof not affected by the asymptotic continuity catastrophe? In short, the reason has to do with the fact that Ref.~\cite{Wilde2018} studies channels rather than states. In finding lower bounds on the cost in this setting, one can simply choose to input a state $RA$ with local finite energy on $R$. Since $R$ is not touched by the channel nor by its simulation, the two relevant states coming into play can be guaranteed to have finite local energy, which enables Ref.~\cite{Wilde2018} to apply the energy-constrained continuity of the entanglement of formation in Ref.~\cite{Shirokov-sq}.
By contrast, in the case of states we are dealing with here, we may impose some energy constraint on the target state $\rho_{AB}^{\otimes n}$, but not on the infinite-dimensional state that is actually prepared by the protocol; thus, we cannot use the energy-constrained continuity in Ref.~\cite{Shirokov-sq} for these states.

In this work, we prove that the entanglement cost of infinite-dimensional states $\rho_{AB}$ is given by the same formula~\eqref{HHT} as in the finite-dimensional case, provided that the quantum entropy of the reduced state of $\rho_{AB}$ on either $A$ or $B$ is finite. For the direct part, we construct an entanglement dilution protocol that works for both finite- and infinite-dimensional systems and achieves a rate as close to the regularized entanglement of formation as desired. To prove that this protocol works as intended, we develop a technique for approximating multiple copies of infinite-dimensional mixed states, which goes beyond a mere application of typicality. The protocol uses one-way LOCC as in the finite-dimensional case, and remarkably, even if $\rho_{AB}$ is infinite-dimensional, it is implementable with only \textit{finitely many} bits of one-way classical communication.

To demonstrate the optimality of this protocol (converse part), we prove that even with the use of separable operations (SEP) one cannot achieve a lower entanglement dilution rate. 
To circumvent the obstacles arising from infinite dimensionality, we show a generalized result on the monotonicity of the entanglement of formation under infinite-dimensional SEP operations, progressing beyond the existing monotonicity results~\cite{PhysRevLett.83.1455,PhysRevLett.84.4781,Vidal2000,Gheorghiu2008}. 
Then, to circumvent the asymptotic continuity catastrophe, we develop an alternative argument based on the semi-continuity bound on the entanglement of formation of infinite-dimensional states recently introduced by Ref.~\cite{Shirokov2022}, rather than the conventional asymptotic continuity in Refs.~\cite{Nielsen2000, tightuniform, Mark2020}, which is applicable only to finite-dimensional states, or the continuity relation in Ref.~\cite{Shirokov-sq}, which instead applies only to two energy-constrained states. 

The rest of this paper is organized as follows.
In Sec.~\ref{sec:notation} we introduce our terminology and establish some notation. In Sec.~\ref{sec:main_results} we present our main results. To prove them, in Sec.~\ref{sec:achievability} we construct the asymptotically optimal one-way LOCC protocol for entanglement dilution to infinite-dimensional states, and in Sec.~\ref{sec:converse} we establish the converse statement, demonstrating that not even SEP operations can do better. 
We conclude by summarizing our findings and describing some open problems in Sec.~\ref{sec:conclusion}.

\section{Preliminaries}
\label{sec:notation}

In this section, we fix our terminology and explain our notation. In Sec.~\ref{sec:quantum_mechanics}, we summarize the formulation of quantum mechanics in infinite-dimensional systems. In Sec.~\ref{sec:topologies}, we discuss various notions of topologies that are needed to define the convergence of infinite-dimensional quantum operations.
In Sec.~\ref{sec:operations}, we introduce several classes of operations considered in our analysis, clearing up some subtleties related to infinite-dimensional systems. In Sec.~\ref{sec:entanglement_cost} and Sec.~\ref{sec:eof}, we formalize the notions of entanglement cost and entanglement of formation. Finally, in Sec.~\ref{subsec:what_goes_wrong}, we discuss why exactly the finite-dimensional analysis of entanglement cost in Ref.~\cite{Hayden-EC} cannot be directly extended to the infinite-dimensional case. Throughout this paper, we let $\N$ denote the set of non-negative integers and $\N_+$ the set of positive integers.

\subsection{Quantum mechanics}
\label{sec:quantum_mechanics}

Quantum systems are represented by complex separable Hilbert spaces $\mathcal{H}$. 
A linear operator $B$ on $\mathcal{H}$ is said to be \deff{bounded} if the \deff{operator norm}
\begin{equation}
\label{eq:operator_norm}
    \left\|B\right\|_\infty\coloneqq\sup\qty{\left\|B\ket{v}\right\|:\ket{v}\in\mathcal{H},\,\left\|\ket{v}\right\|\leq 1}
\end{equation}
is finite, where $\left\|\ket{v}\right\|$ for $\ket{v}\in\mathcal{H}$ is the norm induced by the inner product of $\mathcal{H}$. In what follows, we will denote by $\mathcal{B}\qty(\mathcal{H})$ the Banach space of bounded operators on $\mathcal{H}$ equipped with the operator norm $\|\cdot\|_\infty$. We let $\mathds{1}\in\mathcal{B}\qty(\mathcal{H})$ denote the identity operator. A bounded operator $B\in\mathcal{B}\qty(\mathcal{H})$ is said to be compact if $B$ maps the unit ball of $\mathcal{H}$ into a relatively compact set, i.e., if the closure of $B\qty{\ket{v}\in\mathcal{H}:\left\|\ket{v}\right\|\leq 1}$ is compact in $\mathcal{H}$. The Banach space of compact operators on $\mathcal{H}$ equipped with the operator norm will be denoted by $\mathcal{K}(\mathcal{H})$.

A bounded operator $B\in\mathcal{B}\qty(\mathcal{H})$ is said to be positive semidefinite if $\bra{v} B \ket{v} \geq 0$ for any $\ket{v}\in\mathcal{H}$, which we write as $B\geq 0$.
Let $\{\ket{n}\}_n$ be any orthonormal basis of $\mathcal{H}$.
Given any positive semidefinite bounded operator $B$, the convergence of the (possibly infinite) sum $\sum_n\bra{n} B \ket{n}$, which contains only non-negative terms, as well as its value if it converges, is independent of the choice of $\{\ket{n}\}_n$~\cite[Sec.~18]{conway2000course}. If the sum converges, the trace of the positive semidefinite operator $B$ is defined as $\Tr\qty[B]\coloneqq\sum_n\bra{n} B \ket{n}$.
A bounded operator $T\in \mathcal{B}\qty(\mathcal{H})$ is said to be of \deff{trace class} if $\|T\|_1 \coloneqq \Tr\qty[|T|]=\sum_n\bra{n} |T| \ket{n} < \infty$, where $|T| \coloneqq \sqrt{T^\dag T}\geq 0$, and $T^\dag$ is the adjoint operator of $T$.
Although a trace-class operator is not necessarily positive semidefinite, the above definition of the trace extends naturally to any trace-class operator $T$, by setting
\bb
\label{eq:definition_trace}
    \Tr\qty[T]\coloneqq\sum_n\bra{n} T \ket{n},
\ee
which turns out to converge and is independent of the choice of orthonormal basis $\{\ket{n}\}_n$~\cite[Sec.~18]{conway2000course}.
In the finite-dimensional case, the trace defined in~\eqref{eq:definition_trace} coincides with the sum of the diagonal elements in the matrix representation of the finite-dimensional operator.
The set of trace-class operators is necessarily compact. We will let $\mathcal{T}\qty(\mathcal{H})$ denote the Banach space of trace class operators equipped with the \deff{trace norm} $\|\cdot\|_1$. For two trace-class operators $T,T^\prime\in\mathcal{T}\qty(\mathcal{H})$, we will call the quantity $\frac{1}{2}\left\|T-T^\prime\right\|_1$ the \deff{trace distance} between $T$ and $T^\prime$, and the corresponding topology the \deff{trace-norm topology} on $\mathcal{T}\qty(\mathcal{H})$. We will use the trace-norm topology for $\mathcal{T}\qty(\mathcal{H})$ throughout the paper.

Quantum states are then represented by \deff{density operators}, i.e., positive semidefinite trace-class operators on $\mathcal{H}$ with trace $1$. We let $\mathcal{B}_+(\mathcal{H})$, $\mathcal{K}_+(\mathcal{H})$, and $\mathcal{T}_+(\mathcal{H})$ denote the cones of positive semidefinite bounded, compact, and trace-class operators, respectively. The set of density operators will be denoted instead by $\mathcal{D}\qty(\mathcal{H})\coloneqq\qty{\rho\in\mathcal{T}_+\qty(\mathcal{H}):\,\Tr\qty[\rho]=1}$. Given a state $\rho\in \mathcal{D}\qty(\mathcal{H})$ with spectral decomposition $\rho = \sum_x p(x) \ketbra{x}$, the \deff{quantum entropy} of $\rho$ is defined by
\bb
S\qty(\rho) \coloneqq -\Tr\qty[ \rho \log_2 \rho ] = - \sum_x p(x) \log_2 p(x)\, ,
\label{entropy}
\ee
where we set $0\log_2 0\coloneqq 0$ and $S\qty(\rho) = \infty$ if the series on the right-hand side diverges. 
The entropy function $S$ on $\mathcal{D}\qty(\mathcal{H})$ is: concave, meaning that $S\qty(p\rho+\qty(1-p)\sigma)\geq pS\qty(\rho)+\qty(1-p)S\qty(\sigma)$ for all $\rho,\sigma\in\mathcal{D}\qty(\mathcal{H})$ and all $p\in[0,1]$; and lower semi-continuous, i.e., $\liminf_{\widetilde{\rho}\to\rho} S\qty(\widetilde{\rho})\geq S\qty(\rho)$, where the limit is taken in the trace-norm topology.

The Hilbert space describing a composite quantum system $AB$ is the tensor product of the Hilbert spaces representing its subsystems, in formula $\mathcal{H}_{AB} \coloneqq \mathcal{H}_A\otimes \mathcal{H}_B$.
We may use an operator with subscript $AB$ to represent the state of $AB$, e.g., $\rho_{AB}\in\mathcal{D}\qty(\mathcal{H}_A\otimes \mathcal{H}_B)$. The reduced state of $A$ obtained from the partial trace over $B$ is denoted by $\rho_{A} = \Tr_B\qty[\rho_{AB}]\in\mathcal{D}\qty(\mathcal{H}_A)$, and similarly for $B$. The composite system obtained by joining $n$ copies of $AB$ is indicated by $A^nB^n$.

A quantum state $\sigma_{AB}$ of $AB$ is said to be \deff{separable} if it is in the closed convex hull (in the trace-norm topology) of the set of all tensor-product states $\omega_A\otimes\tau_B$ for $\omega_A\in\mathcal{D}\qty(\mathcal{H}_A)$ and $\tau_B\in\mathcal{D}\qty(\mathcal{H}_B)$. The set of separable states will be denoted by $\SEP(A\!:\!B)$; it is formally defined by
\bb
\SEP(A\!:\!B) \coloneqq \cl \co \left\{ \omega_A\otimes\tau_B:\ \omega_A\in\mathcal{D}\qty(\mathcal{H}_A),\ \tau_B\in\mathcal{D}\qty(\mathcal{H}_B) \right\} ,
\label{separable_states}
\ee
where $\cl$ denotes the closure with respect to the trace-norm topology. Rather intuitively, $\SEP(A\!:\!B)$ can alternatively be defined as the set of barycenters of all Borel probability measure $\mu$ on $\mathcal{D}\qty(\mathcal{H}_A)\times\mathcal{D}\qty(\mathcal{H}_B)$; that is, any separable state $\sigma_{AB}$ can be written in terms of the (Bochner) integral as $\sigma_{AB} = \int \dd\mu \big(\omega_A,\tau_B\big)\ \omega_A\otimes\tau_B$~\cite{Holevo2005,Holevo-CJ-arXiv}. If there exists a representation of the form $\sigma_{AB} = \sumno_{x=1}^{\infty} p(x)\, \omega_A^{(x)} \otimes\tau_B^{(x)}$ then $\sigma_{AB}$ is said to be \deff{countably separable}. In the infinite-dimensional case, it is known that separable states that are not countably separable do exist~\cite{Holevo2005,Holevo-CJ-arXiv}.

A state that is not separable is said to be \deff{entangled}~\cite{Werner, Holevo2005}. The simplest example of an entangled state is the two-qubit maximally entangled state, called an \deff{entanglement bit} or \deff{ebit}. This is defined by $\Phi \coloneqq \ketbra{\Phi}$, where 
\begin{equation}
\label{eq:ebit}
    \ket{\Phi}_{AB}\coloneqq \frac{1}{\sqrt2} \left(\ket{0}_A\otimes\ket{0}_B + \ket{1}_A\otimes\ket{1}_B\right)\in\mathcal{H}_A\otimes\mathcal{H}_B\, .
\end{equation}

\subsection{Topologies on the set of quantum channels} \label{sec:topologies}

We use a subset of linear maps $\pazocal{E}:\mathcal{T}\qty(\mathcal{H})\to\mathcal{T}\qty(\mathcal{H}^\prime)$ to represent transformations between quantum states. Here, $\mathcal{H}$ and $\mathcal{H}^\prime$ are the input and output Hilbert spaces, respectively. The linear map $\pazocal{E}$ is called: \deff{positive}, if $\pazocal{E}\qty(T)\geq 0$ for all $T\geq 0$; \deff{completely positive} (CP), if $I_k\otimes \pazocal{E}$ is positive for all $k\in \N_+$, where $I_k$ denotes the identity map on the set of $k\times k$ complex matrices; and \deff{trace-preserving} (TP), if $\Tr\qty[\pazocal{E}\qty(T)]=\Tr\qty[T]$ for all $T\in\mathcal{T}\qty(\mathcal{H})$. A completely positive trace-preserving (CPTP) map is also called a \deff{quantum channel}. Quantum channels describe the transformations between quantum states allowed within the law of quantum mechanics. For a finite or countably infinite set $\mathcal{J}$, a family of CP maps $\qty(\pazocal{E}_j)_{j\in\mathcal{J}}$ summing up a TP map $\sum_{j\in\mathcal{J}} \pazocal{E}_j$ is called an instrument. Instruments describe physical measurements on quantum systems: the output state corresponding to input state $\rho$ and measurement outcome $j$ is simply $\frac{\pazocal{E}_j(\rho)}{\Tr \pazocal{E}_j(\rho)}$.

Since we are dealing with possibly infinite-dimensional systems, we need to choose a topology on the space of quantum channels. A possible choice is the topology of uniform convergence, induced by the diamond norm~\cite{Aharonov1998}. However, it has been observed that this topology is often too fine, and as a consequence, it exhibits some pathological behaviors---for instance, it places physically close channels at nonzero distance~\cite[Sec.~II]{VV-diamond}. Two alternatives that are more physically sound are the topology induced by the \emph{energy-constrained diamond norm}~\cite{Shirokov2018, VV-diamond, PLOB} and the topology of \emph{strong convergence}. Fortunately, Ref.~\cite[Proposition~3]{Shirokov2018} has shown that for many physically motivated Hamiltonians, these two alternatives are equivalent, so we adopt the latter without loss of generality. We then say that a sequence\footnote{We should talk more generally about convergences of nets rather than sequences. However, since we are working with separable Hilbert spaces, the topology of strong convergence turns out to be \emph{metrizable}, i.e., induced by a metric. This simple fact, stated in Ref.~\cite{Shirokov2018}, can be proved by considering a dense set $(\rho_n)_{n\in \N}$ of states, and then defining the distance $d\big(\EE,\pazocal{E}'\big) \coloneqq \sum_{n\in \N} 2^{-n} \left\|\left(\EE - \EE' \right)(\rho_n)\right\|_1$. Then, in a metric space, topologies can be defined by specifying the converging sequences instead of the converging nets. The details are left to the reader.} $\big(\EE_n\big)_n$ of CP maps $\EE_n: \mathcal{T}\qty(\mathcal{H}) \to \mathcal{T}\qty(\mathcal{H}^\prime)$ converges \deff{strongly} to a map $\EE: \mathcal{T}\qty(\mathcal{H}) \to \mathcal{T}\qty(\mathcal{H}^\prime)$, and we write $\EE_n \tendsn{s} \EE$, if
\bb
\lim_{n\to\infty} \left\| \EE_n(T) - \EE(T)\right\|_1 = 0 \qquad \forall\ T\in \mathcal{T}\qty(\mathcal{H})\, .
\ee
Under these assumptions, this map $\EE$ must be CP\@. If the $\EE_n$ are quantum channels, then so must be $\EE$. 
The closure with respect to the topology of strong convergence is denoted by $\cl_{\mathrm{s}}(\cdot)$.

So far we have discussed the strong convergence of quantum channels. However, for what follows, we will also care about strong convergence of \emph{instruments}, i.e., countable collections $\mathcal{I} = (\EE_i)_i$ of CP maps $\EE_i:\mathcal{T}\qty(\mathcal{H}) \to \mathcal{T}\qty(\mathcal{H}^\prime)$ that sum up to a CPTP map $\sum_i \EE_i$.
One can easily think of instruments as quantum channels by appending an ancillary output system that records the outcome of the measurement, i.e., by constructing the quantum channel $\rho \mapsto \sum_i \EE_i(\rho) \otimes \ketbra{i}$. In this way, we can extend the notion of strong convergence from channels to instruments.
We let  $\mathcal{I}_\nu \tends{s}{\nu\to\infty} \mathcal{I}$ denote the strong convergence of the sequence of instruments $\big(\mathcal{I}_\nu\big)_\nu$ to $\mathcal{I}$. 

\subsection{LOCC for infinite-dimensional systems}
\label{sec:operations}

To talk about entanglement manipulation, we need to first discuss a few subtleties associated with the definition of local operations and classical communication (LOCC) in the case of infinite-dimensional systems. The good news is that most of the work done by Ref.~\cite[Sec.~2.2]{LOCC} to define finite-dimensional LOCC carries over straight away to the infinite-dimensional case. Let us summarise the main steps of Ref.~\cite[Sec.~2.2]{LOCC} here. For some positive integer $N$, let $\mathcal{H}=\mathcal{H}_{A_1}\otimes\cdots\otimes\mathcal{H}_{A_N}$ represent an $N$-partite quantum system, where $\mathcal{H}_{A_n}$ for each $n\in\qty{1,\ldots,N}$ is the separable Hilbert space associated with the $n^{\text{th}}$ party. 
For simplicity of presentation, here we use the same input and output systems $\mathcal{H}$ in the following definitions of maps, but the general case with different input and output systems follows similarly. 
We allow for either finite or countably infinite sets as index sets of the instruments in the following definitions.
For two instruments $\mathcal{I}=\qty(\pazocal{E}_i)_{i\in I}$ and $\mathcal{I}^\prime=\big(\pazocal{E}_{j}^\prime\big)_{j\in J}$, we call $\mathcal{I}^\prime$ a coarse-graining of $\mathcal{I}$ if there exists a $J$-indexed partition of $I$, i.e., $I = \coprod_{j\in J} I_j$, such that for all $j\in J$, it holds that $\pazocal{E}_{j}^\prime = \sum_{i\in I_{j}} \pazocal{E}_i$.
Despite the somewhat cumbersome definition, the physical meaning is clear: we can implement $\mathcal{I}^\prime$ by first measuring $\mathcal{I}$, and then ``binning'' the outcome into one of the subsets $I_{j}$. 
For each party $n$, an instrument $\mathcal{I}=\qty(\pazocal{E}_i)_{i}$ is called one-way local with respect to party $n$ if each of its CP maps has the form
\bb
\pazocal{E}_i=\qty(\bigotimes\nolimits_{n^\prime\neq n}\pazocal{F}_i^{\qty(n^\prime)})\otimes \pazocal{E}_i^{\qty(n)}\, , 
\ee
where $\pazocal{E}_i^{\qty(n)}$ is a CP map on $\mathcal{T}\qty(\mathcal{H}_{A_n})$, and $\pazocal{F}_i^{\qty(n^\prime)}$ is a CPTP map on $\mathcal{T}\qty(\mathcal{H}_{A_{n^\prime}})$. The interpretation here is that party $n$ applies an instrument $\big(\pazocal{E}_i^{\qty(n)}\big)_i$ on their system, communicates the outcome $i$ to all the other parties, and each party $n'\neq n$ post-processes their system with some quantum channel $\pazocal{F}_i^{\qty(n^\prime)}$ that may depend on $i$.

An instrument $\mathcal{I}^\prime$ is LOCC-linked to $\qty(\pazocal{E}_i)_i$ if, for all $i$, there exists a one-way local instrument $\big(\pazocal{F}_{j|i}\big)_{j}$ with respect to some party $n_i\in\qty{1,\ldots,N}$ such that $\mathcal{I}^\prime$ is a coarse-graining of the instrument $\big(\pazocal{F}_{j|i}\circ\pazocal{E}_i\big)_{i,j}$. Physically, after the outcome $i$ has been obtained and broadcast to all parties, one of them, say party $n_i$, depending on $i$, can, in turn, apply another instrument on their system. We say that:
\begin{enumerate}
    \item $\mathcal{I}\in\locc_1$, if $\mathcal{I}$ is the coarse-graining of a one-way local instrument with respect to some party;
    \item $\mathcal{I}\in\locc_r$ ($r\geq 2$), if $\mathcal{I}$ is LOCC-linked to some $\mathcal{I}\in\locc_{r-1}$;
    \item $\mathcal{I}\in\locc_\mathbb{N}$, if $\mathcal{I}\in\locc_r$ for some finite $r$.
\end{enumerate}

\begin{note}
Given an instrument in $\locc_1$, we may be interested in specifying with respect to which party it is one-way local up to coarse-graining. If the underlying system is bipartite ($N=2$), we identify the two parties with Alice and Bob, and we use the notation $\locc_\to$ and $\locc_\leftarrow$ to denote the sets of instruments that are (up to coarse-graining) one-way local with respect to Alice and Bob, respectively. Clearly,
\bb
\locc_1 = \locc_\to \cup \locc_\leftarrow\, .
\ee
\end{note}

We are now ready to extend the definition of LOCC in Ref.~\cite[Sec.~2.2]{LOCC} from finite- to infinite-dimensional systems.

\begin{Def}[(LOCC)] \label{def:LOCC}
An instrument $\mathcal{I} = \big(\pazocal{E}_i\big)_i$ on some $N$-partite separable Hilbert space $\mathcal{H}=\mathcal{H}_{A_1}\otimes\cdots\otimes\mathcal{H}_{A_N}$ is said to be \deff{LOCC} if there exists a sequence of instruments $\mathcal{I}_1, \mathcal{I}_2,\ldots \in \locc_\mathbb{N}$ such that $\mathcal{I}_{\nu+1}$ is LOCC-linked to $\mathcal{I}_{\nu}$ for all $\nu\geq 1$, and each $\mathcal{I}_\nu$ has a coarse-graining $\mathcal{I}'_\nu = \big(\pazocal{E}^{(\nu)}_i\big)_i$ with the property that
\bb
\mathcal{I}'_\nu \tends{s}{\nu\to\infty} \mathcal{I}
\ee
strongly, in the sense of Sec.~\ref{sec:topologies}.
\end{Def}

To prove the converse part of our main result, we need to extend also the notion of \emph{separable operations}, originally introduced in Refs.~\cite{Rains1997, Vedral1998} for the finite-dimensional case, to infinite-dimensional systems.

\begin{Def}[(Countably separable operations and separable operations)] \label{def:separable_channels}
A quantum channel $\EE$ acting on an $N$-partite separable Hilbert space $\mathcal{H}=\mathcal{H}_{A_1}\otimes\cdots\otimes\mathcal{H}_{A_N}$ is said to be \deff{countably separable}, and we write $\EE\in \sep_\N$, if there exists a Kraus representation $\pazocal{E}\qty(\rho)=\sum_{\alpha=1}^\infty K_{\alpha}^{\vphantom{\dag}} \rho K_{\alpha}^\dag$, where all Kraus operators $K_{\alpha}$ can be written as tensor products $K_{\alpha} = \bigotimes_{n=1}^{N} K_{\alpha}^{(n)}$, with each $K_{\alpha}^{(n)}$ acting on the corresponding local space $\mathcal{H}_{A_n}$. The set of \deff{separable channels}, denoted $\sep$, is the closure of $\sep_\N$ with respect to the topology of strong convergence, i.e.,
\bb
\sep \coloneqq \cl_{\mathrm{s}}\left(\sep_\N\right) .
\label{separable_channels}
\ee
We say that an instrument $(\EE_i)_i$ on $\mathcal{H}$ is separable (respectively, countably separable) if the corresponding channel $\rho \mapsto \sum_i \EE_i(\rho) \otimes \ketbra{i}_{X}$ is separable (respectively, countably separable), where $X$ is a classical system assigned to one of the parties, without loss of generality $A_1$. (Since classical information can be copied, it is immaterial which party gets $X$.)
\end{Def}

The rationale behind the above definition reflects the phenomenology observed for separable states in infinite-dimensional systems. According to~\eqref{separable_states}, $\SEP(A\!:\!B)$ is defined as the closure of the set of countably separable states with respect to an appropriate topology for states; in the same way, we now construct $\sep$ as the closure of $\sep_\N$ with respect to an appropriate quantum channel topology. The fact that taking the closure is a necessary step will be apparent from the discussion below.

There is a simple hierarchy among the above classes of operations (or instruments), namely,
\bb
\locc_\to,\,\locc_\leftarrow \subset \locc_1\subset \locc_2 \subset \ldots \subset \locc_\N \subset \locc \subset \overline{\locc} \subset \sep\, ,
\label{eq:infinite_operations}
\ee
where $\overline{\locc}\coloneqq\cl_{\mathrm{s}} \left(\locc_\N\right)$.
Note that $\cl_{\mathrm{s}} \left(\locc_\N\right)=\cl_{\mathrm{s}} \left(\locc\right)$ by definition of LOCC\@.
The inclusion relation~\eqref{eq:infinite_operations} mirrors the finite-dimensional one in Ref.~\cite[Eq.~(3)]{LOCC}. One of the results of Ref.~\cite{LOCC} is that all inclusions in~\eqref{eq:infinite_operations} are strict already in finite dimension. In general systems, we have
\bb
\locc_\N \subset \sep_\N \subset \sep\, .
\label{eq:inclusions_operations_infinite_dim}
\ee
Already in finite dimension $\locc_\N\neq\sep_\N=\sep$, so that the first inclusion in~\eqref{eq:inclusions_operations_infinite_dim} is in general strict. As for the second, it is also strict in the case of infinite-dimensional systems. A simple example that suffices to show this claim is as follows: consider a ``replacer channel'' acting on a bipartite system $AB$
\bb
\EE_{\sigma} (\rho_{AB}) = \sigma_{AB} \Tr\qty[\rho_{AB}]\,,
\label{replacer_channel}
\ee
where $\sigma=\sigma_{AB}$ is a separable but not countably separable state. As it turns out, 
\bb
\EE_{\sigma} \in \locc \setminus \sep_\N \subseteq \sep \setminus \sep_\N\, .
\label{LOCC_SEPN_gap_replacer_channel}
\ee
The fact that $\EE_\sigma \notin \sep_\N$ can be seen by observing that if $\rho$ is a product state, the output of any operation in $\sep_\N$ acting on $\rho$ must be a countably separable state, but $\sigma$ is not countably separable. The fact that $\EE_\sigma \in \locc$ is proved as follows. Pick a sequence of countably separable states $\sigma_\nu$ with $\lim_{\nu\to\infty} \|\sigma - \sigma_\nu\|_1 = 0$. This is possible because the set of separable states is the closure of that of countably separable states with respect to the trace-norm topology. Then, construct the replacer channels
\bb
\EE_\nu :\ \rho \mapsto \sigma_\nu \Tr\qty[\rho]\, .
\ee
Note that $\EE_\nu$ is one-way local with respect to any party. Moreover, $\EE_{\nu+1}\circ \EE_\nu = \EE_{\nu+1}$. Hence, $\EE_{\nu+1}$ is LOCC-linked to $\EE_\nu$ for all $\nu$. Since $\EE_\nu \tends{s}{\nu\to\infty} \EE_\sigma$, we have $\EE_\sigma\in \locc$.

In view of~\eqref{LOCC_SEPN_gap_replacer_channel}, we see that the strong closure in~\eqref{separable_channels} is needed to include in the definition of SEP some natural channels that we would want to call separable. As it happens for states, it is natural to conjecture that SEP channels that are not in $\sep_\N$ still admit a continuous Kraus decomposition with tensor-product Kraus operators. We leave the investigation of this question, which is not essential for our problem, for future work.

\subsection{Entanglement cost}
\label{sec:entanglement_cost}

The two fundamental ingredients of entanglement manipulation are entanglement distillation and entanglement dilution, two essentially opposite tasks~\cite{Bennett-distillation, Bennett-distillation-mixed, Bennett-error-correction}.
The former consists of transforming many IID copies of a given state $\rho_{AB}$ into as many ebits as possible while employing only LOCC (or their variants as discussed in Sec.~\ref{sec:operations}) and achieving an asymptotically vanishing error.
The converse task of entanglement dilution, which will be of interest here, consists of preparing many copies of $\rho_{AB}$ with LOCC while consuming as few ebits as possible, and again with vanishing error. 

More formally, the task of entanglement dilution involves two parties, Alice and Bob, who can perform a restricted family $\mathcal{O}$ of CPTP maps, assisted by initially shared entanglement in the form of ebits. For the sake of this paper, it suffices to consider the cases
\bb
\mathcal{O}\in \left\{ \locc_\to,\, \locc_\leftarrow,\,  \locc_\N,\, \locc,\, \overline{\locc},\, \sep_\N,\, \sep \right\} .
\label{eq:classes}
\ee
Let $\rho = \rho_{AB}\in\mathcal{D}\qty(\mathcal{H}_A\otimes\mathcal{H}_B)$ be a target quantum state on a (either finite- or infinite-dimensional) bipartite quantum system $AB$. The \deff{entanglement cost} of $\rho$ under operations in $\mathcal{O}$ is defined by
\bb
\label{eq:entanglement_cost}
E_{c,\, \mathcal{O}} (\rho) \coloneqq \inf\left\{ r >0 :\ \lim_{n\to\infty} \inf_{\EE_n \in \mathcal{O}}\frac12\left\| \EE_n\big(\Phi^{\otimes \floor*{rn}}\big) - \rho^{\otimes n} \right\|_1 = 0\right\} ,
\ee
where $\Phi = \ketbra{\Phi}$ is the ebit defined by~\eqref{eq:ebit}, and for a fixed $n$, the operation $\EE_n$ acts on a bipartite system composed of $\floor*{rn}+\floor*{rn}$ qubits and outputs a state in $A^nB^n$. In~\eqref{eq:entanglement_cost}, $r$ should be interpreted as an \deff{achievable rate}, i.e., as an achievable ebit cost per copy of $\rho$ to be prepared. If the set on the right-hand side of~\eqref{eq:entanglement_cost} is empty, then we set $E_{c,\,\mathcal{O}}(\rho) = +\infty$.

\begin{rem} \label{inclusion_hierarchy_cost_rem}
Given two classes of operations  $\mathcal{O}$ and $\mathcal{O}'$, if $\mathcal{O}\subseteq \mathcal{O}'$ then
\bb
E_{c,\, \mathcal{O}} (\rho) \geq E_{c,\, \mathcal{O}'} (\rho)
\ee
for all bipartite states $\rho$. 
\end{rem}

\begin{rem} \label{common_randomness_rem}
All the classes of free operations in~\eqref{eq:classes} are closed under composition, implying that
\bb
E_{c,\,\mathcal{O}}\big(\EE(\rho)\big) \leq E_{c,\,\mathcal{O}}(\rho) \qquad \forall\ \EE\in \mathcal{O}\, .
\label{data_processing_cost}
\ee
Also, Alice and Bob can create arbitrary countably separable states with even the smallest of the classes in~\eqref{eq:classes}, namely, either $\locc_\to$ or $\locc_\leftarrow$\@. It then follows that
\bb
E_{c,\, \mathcal{O}}(\rho) = E_{c,\, \mathcal{O}}(\rho\otimes \sigma)
\ee
for all countably separable states $\sigma$ and for all $\mathcal{O}$ in~\eqref{eq:classes}. In particular, granting Alice and Bob access to a source of (discrete) shared randomness does not alter the entanglement cost.
\end{rem}

\begin{rem} \label{closed_under_parallel_composition_rem}
All the classes of free operations in~\eqref{eq:classes} are closed under parallel use represented by tensor product, meaning that for all $\mathcal{O}$ as in~\eqref{eq:classes}, if $\EE_1 \in \mathcal{O}_{A_1:B_1\to A'_1:B'_1}$ and $\EE_2 \in \mathcal{O}_{A_2:B_2\to A'_2:B'_2}$, where we specified as subscripts the input and output systems, then $\EE_1\otimes \EE_2 \in \mathcal{O}_{A_1A_2:B_1B_2\to A'_1A'_2:B'_1B'_2}$. As an immediate consequence, for all pairs of states $\rho = \rho_{AB}$ and $\omega = \omega_{A'B'}$ it holds that
\bb
E_{c,\,\mathcal{O}}(\rho \otimes \omega) \leq E_{c,\,\mathcal{O}}(\rho) + E_{c,\,\mathcal{O}}(\omega)\, .
\ee
\end{rem}

In the definition of entanglement cost~\eqref{eq:entanglement_cost}, the preparation of $n$ copies of $\rho$ can incur a nonzero error, however such error must vanish asymptotically as $n\to\infty$. Because of the fact that manipulation is not required to be zero-error, the entanglement cost under any class of operations and under its strong closure turns out to be the same. In particular, the cost under $\locc$ and that under $\locc_\N$ coincide. This simple fact is proved below. 

\begin{lemma}[(Entanglement cost under a class of operations and under its closure)] \label{same_cost_lemma}
For all classes of operations $\mathcal{O}$ (not necessarily as in~\eqref{eq:classes}) and all bipartite states $\rho = \rho_{AB}$, it holds that
\bb
E_{c,\,\mathcal{O}}(\rho) = E_{c,\,\widebar{\mathcal{O}}}(\rho)\, ,
\label{same_cost_O_and_O_closure}
\ee
where $\widebar{\mathcal{O}} \coloneqq \cl_{\mathrm{s}} \left( \mathcal{O} \right)$ denotes the closure of $\mathcal{O}$ with respect to the strong operator topology. In particular,
\bb
E_{c,\, \locc_\N}(\rho) = E_{c,\, \locc}(\rho)\, ,\qquad E_{c,\, \sep_\N}(\rho) = E_{c,\, \sep}(\rho)\, .
\label{same_cost_finite_rounds}
\ee
\end{lemma}

\begin{proof}
Since $\mathcal{O}\subseteq \widebar{\mathcal{O}}$, by Remark~\ref{inclusion_hierarchy_cost_rem}, we have that $E_{c,\,\mathcal{O}}(\rho) \geq E_{c,\,\widebar{\mathcal{O}}}(\rho)$. For the opposite inequality, assume that $r>0$ is an achievable rate for entanglement dilution to $\rho$ under operations in $\widebar{\mathcal{O}}$. This means that there exists a sequence of operations $\pazocal{E}_n\in \widebar{\mathcal{O}}$ such that 
\bb
\lim_{n\to\infty} \frac12 \left\| \pazocal{E}_n\big(\Phi^{\otimes \floor*{rn}}\big) - \rho^{\otimes n} \right\|_1 = 0\, .
\ee
Now, because of the definition of strong topology (Sec.~\ref{sec:topologies}), for all $n$, we can find some $\pazocal{E}'_n\in \mathcal{O}$ such that
\bb
\frac12\left\| \pazocal{E}'_n\big(\Phi^{\otimes \floor*{rn}}\big) - \pazocal{E}_n\big(\Phi^{\otimes \floor*{rn}}\big) \right\|_1 \leq \frac1n\, .
\ee
Consequently, we also have that
\bb
\frac12 \left\| \pazocal{E}'_n\big(\Phi^{\otimes \floor*{rn}}\big) - \rho^{\otimes n} \right\|_1 \leq \frac12\left\| \pazocal{E}'_n\big(\Phi^{\otimes \floor*{rn}}\big) - \pazocal{E}_n\big(\Phi^{\otimes \floor*{rn}}\big) \right\|_1 + \frac12 \left\| \pazocal{E}_n\big(\Phi^{\otimes \floor*{rn}}\big) - \rho^{\otimes n} \right\|_1 \tendsn{} 0\, .
\ee
This implies that $r$ is also an achievable rate for entanglement dilution under $\mathcal{O}$, i.e., $r\geq E_{c,\,\mathcal{O}}(\rho)$. Taking the infimum over $r$ shows that $E_{c,\,\widebar{\mathcal{O}}}(\rho) \geq E_{c,\,\mathcal{O}}(\rho)$, completing the proof of~\eqref{same_cost_O_and_O_closure}. The second identity in~\eqref{same_cost_finite_rounds} follows immediately by taking $\mathcal{O} = \sep_\N$. To deduce the first, simply write
\bb
E_{c,\, \locc}(\rho) \leq E_{c,\, \locc_\N}(\rho) = E_{c,\, \overline{\locc}\,}(\rho) \leq E_{c,\, \locc}(\rho)\, ,
\ee
where the first and last inequality follow from Remark~\ref{inclusion_hierarchy_cost_rem} due to~\eqref{eq:infinite_operations}, while the middle equality is simply~\eqref{same_cost_O_and_O_closure} applied to $\mathcal{O} = \locc_\N$.
\end{proof}

Putting together~\eqref{eq:infinite_operations}, Remark~\ref{inclusion_hierarchy_cost_rem}, and Lemma~\ref{same_cost_lemma}, we deduce that
\begin{equation}
\label{eq:different_entanglement_costs}
    E_{c,\, \locc_1} (\rho)\geq E_{c,\,\locc_2} (\rho)\geq \cdots\geq E_{c,\,\locc_\N}(\rho) = E_{c,\,\locc}(\rho) \geq  E_{c,\, \sep_\N}(\rho) = E_{c,\, \sep}(\rho)
\end{equation}
for arbitrary bipartite states $\rho=\rho_{AB}$. In finite dimension, it is known that all of the above notions of cost, which are a priori different, actually coincide. Our results will show that, perhaps surprisingly, the same happens in infinite-dimensional systems.

\subsection{Entanglement of formation}
\label{sec:eof}

A foundational result in entanglement theory is the characterization of entanglement cost under finite-dimensional LOCC as the regularization of the entanglement of formation~\cite{Hayden-EC, Buscemi2011}, one of the most important entanglement measures. 
In a finite- or infinite-dimensional quantum system, the \deff{entanglement of formation} is constructed as the continuous convex-roof extension of the entanglement entropy, i.e.~\cite{Bennett-error-correction, Adam_W_Majewski_2002, Shirokov2010}
\bb
E_f\qty(\rho) \coloneqq \inf_{\rho_{AB} = \int \dd\mu(\psi)\, \psi_{AB}} \int \dd\mu(\psi)\, S\big(\psi_A \big)\, ,
\label{eof}
\ee
where the infimum on the right-hand side is taken over all Borel probability measures $\mu$ on the set of pure states $\psi = \psi_{AB}=\ketbra{\psi}_{AB}$ of $\mathcal{H}_A\otimes \mathcal{H}_B$ having as barycenter $\rho_{AB}$. 
Since pure states on $AB$ have the same local entropy on $A$ and on $B$, in~\eqref{eof}, we can equivalently replace $S\qty(\psi_A)$ with $S\qty(\psi_B)$. Note that Refs.~\cite{Adam_W_Majewski_2002, Shirokov2008-b, Shirokov2010} (see also Ref.~\cite[Sec.~4.4.2]{Shirokov-review} for a recent review) have also introduced another possible definition of entanglement of formation, where one replaces the integral in~\eqref{eof} with a discrete, countably infinite sum:
\bb
\widetilde{E}_f\qty(\rho) \coloneqq \inf_{\rho_{AB} = \sum_x p(x)\, \psi_x^{AB}} \sum_x p(x)\, S\big(\psi_x^A \big)\, .
\label{EoF_tilde}
\ee
Clearly, 
\bb
E_f(\rho) \leq \widetilde{E}_f\qty(\rho)
\label{EoF_smaller_EoF_tilde}
\ee
holds in general, by restricting the measure $\mu$ in~\eqref{eof} to be a countable convex combination of Dirac measures. If $AB$ is finite dimensional, it follows from Carath\'eodory's theorem~\cite[Lemma~A.2]{Uhlmann1998} that in fact $E_f = \widetilde{E}_f$. The same is true for infinite-dimensional systems, provided that one considers only states with finite local entropy; that is, as shown in Ref.~\cite{Shirokov2008-b},
\begin{equation}
\label{eq:min_S_A_B}
E_f(\rho_{AB}) = \widetilde{E}_f(\rho_{AB})\qquad \forall\ \rho_{AB}\in \mathcal{D}\qty(\mathcal{H}_A\otimes \mathcal{H}_B)\quad\text{with}\ \ \min\qty{S\qty(\rho_{A}), S\qty(\rho_{B})}<\infty\, .
\end{equation}
It is presently unknown whether the above two versions of the entanglement of formation, $E_f$ and $\widetilde{E}_f$, coincide in general or not~\cite[Sec.~4.4.2]{Shirokov-review}. However, the continuous version $E_f$ has superior proven analytical properties, which is why it may be regarded as the most appropriate definition.\footnote{We thank Maksim E.\ Shirokov for bringing this point to our attention.} In particular, in Ref.~\cite[Sec.~5]{Shirokov2010} it is proved that:
\begin{itemize}
\item $E_f$ is everywhere lower semi-continuous, meaning that
\bb
\liminf_{n\to\infty} E_f(\rho_n) \geq E_f(\rho)
\label{lsc_EoF}
\ee
for all converging sequences $(\rho_n)_n$ of states $\rho_n \in \mathcal{D}\qty(\mathcal{H}_A \otimes \mathcal{H}_B)$ with limit $\rho\in \mathcal{D}\qty(\mathcal{H}_A \otimes \mathcal{H}_B)$, i.e., such that $\left\|\rho_n - \rho\right\|_1 \tendsn{} 0$. The same property for $\widetilde{E}_f$ is known to hold on sufficiently regular sequences but has not been established in general.
\item $E_f$ is monotonic under operations in $\overline{\locc}$, meaning that
\bb
E_f\big(\EE(\rho)\big) \leq E_f(\rho)\qquad \forall\ \rho = \rho_{AB} \in \mathcal{D}\qty(\mathcal{H}_A \otimes \mathcal{H}_B)\, , \qquad \forall\ \EE\in \overline{\locc}\, .
\label{monotonicity_locc_EoF}
\ee
Note that Ref.~\cite[Sec.~5]{Shirokov2010} only states~\eqref{monotonicity_locc_EoF} explicitly for selective local instruments, but by composing them, taking limits, and using the lower semi-continuity~\eqref{lsc_EoF}, one can establish~\eqref{monotonicity_locc_EoF}. In this work, we will extend~\eqref{monotonicity_locc_EoF} to all separable operations (see Proposition~\ref{prop:monotonicity_entanglement_formation_under_sep} below).
\end{itemize}
Additionally, the restriction in~\eqref{eq:min_S_A_B} immediately implies that the entanglement of formation itself is finite, i.e., $E_f\qty(\rho)<\infty$, due to the concavity of the entropy. In fact, 
$\int \dd\mu(\psi)\, S\big(\psi_A \big) \leq S\left(\int \dd\mu(\psi)\, S\big(\psi_A \big)\right) = S\qty(\rho_{A})$, and analogously for $B$. In this work, we will always deal with states having finite local entropy, hence the distinction between $E_f$ and $\widetilde{E}_f$ is not crucial for us. 

Another important property of the entanglement of formation is subadditivity; for every two states $\rho_{AB}$ and $\sigma_{A'B'}$, we have that
\begin{equation}
    E_f\big(\rho_{AB} \otimes \sigma_{A'B'}\big)\leq E_f\big(\rho_{AB}\big) + E_f\big(\sigma_{A'B'}\big),
\end{equation}
where it is understood that on the left-hand side, the entanglement of formation is computed with respect to the bipartition $AA'$ versus $BB'$.
Note that the additivity $E_f\big(\rho_{AB} \otimes \sigma_{A'B'}\big)= E_f\big(\rho_{AB}\big) + E_f\big(\sigma_{A'B'}\big)$ does not hold in general~\cite{Hastings2008, Shor2004}. When combined with Fekete's lemma~\cite{Fekete1923}, subadditivity implies that the \deff{regularized entanglement of formation} exists and can be evaluated as
\bb
E_f^\infty\qty(\rho) \coloneqq \lim_{n\to\infty} \frac1n\, E_f\big(\rho^{\otimes n}\big) = \inf_{n\in \N_+} \frac1n\, E_f\big(\rho^{\otimes n}\big) .
\label{regularized_eof}
\ee
Again, it is understood that the bipartition used on the right-hand sides of~\eqref{regularized_eof} is $A^n$ versus $B^n$. The main result of Ref.~\cite{Hayden-EC} is precisely that
\bb
E_{c,\,\locc}(\rho) = E_f^\infty\qty(\rho)
\label{Hayden_et_al_result}
\ee
for all states $\rho_{AB}\in \mathcal{D}\big(\mathcal{H}_A\otimes \mathcal{H}_B\big)$ with both $\dim \mathcal{H}_A < \infty$ and $\dim \mathcal{H}_B < \infty$, provided that one restricts also the LOCC operations to be finite-dimensional.

In the finite dimension, the entanglement of formation is not only lower semi-continuous but, in fact, asymptotically continuous. Namely, for any two states $\rho_{AB}$ and $\omega_{AB}$ of a bipartite system $AB$ with finite minimum local dimension $d \coloneqq \min\qty{\dim\mathcal{H}_A,\dim\mathcal{H}_B}$, it holds that~\cite[Corollary~2]{Mark2020}
\bb
\left| E_f\qty(\rho) - E_f\qty(\omega) \right| \leq \sqrt{\e(2-\e)} \log_2 (d-1) + h_2\qty(\sqrt{\e(2-\e)})\, ,
\label{asymptotic_continuity_EoF}
\ee
where $h_2(x)\coloneqq -x\log_2 x - (1-x) \log_2(1-x)$ is the binary entropy function, and 
\begin{equation}
\e\coloneqq \frac12 \left\|\rho_{AB} - \omega_{AB}\right\|_1
\end{equation}
is the trace distance between $\rho$ and $\omega$ (cf.\ Refs.~\cite{Nielsen2000, tightuniform}), assumed to satisfy the inequality $0\leq \e \leq 1-\frac{\sqrt{2d-1}}{d}$. The above inequality implies not only the continuity of $E_f$ but also a special form of uniform continuity that features a linear dependence on $\log_2 d$ --- this is precisely what is commonly referred to as asymptotic continuity. 

\subsection{What exactly goes wrong in the infinite-dimensional case?} \label{subsec:what_goes_wrong}

Before presenting our main result, it is worth discussing why exactly the finite-dimensional proof of~\eqref{Hayden_et_al_result} in Ref.~\cite{Hayden-EC} cannot be directly extended to infinite-dimensional systems. We already gave a broad overview of the problems in Sec.~\ref{subsec_main_contribution}, but now we can be more technically precise and pinpoint where the issues lie. Essentially, there are four main problems, two in the direct part and two in the converse.
\begin{enumerate}[(1)]

\item 
\label{issue:1}
The direct part of the argument in Ref.~\cite{Hayden-EC} relies crucially on the existence of an entanglement dilution protocol that prepares any pure state $\psi_{AB}$ by consuming ebits at a rate equal to $S(\psi_A)$. The first problem lies precisely in how to construct such a protocol in the infinite-dimensional case. In fact, in order for the cost to be identified with the local quantum entropy, we need to apply some notion of typicality to the locally reduced state $\psi_A$. How does this work beyond finite dimension? We will solve this issue by resorting to the fact that the notion of weak typicality is still applicable to countably infinite ensembles.

\item 
\label{issue:2}
Another more serious problem that arises in attempting to generalize the achievability (direct) part of the argument in Ref.~\cite{Hayden-EC} is as follows: given any pure-state decomposition $\rho_{AB} = \sum_x p(x)\, \psi^{(x)}_{AB}$, instead of $\rho_{AB}^{\otimes n}$, we aim to prepare its typical part, given, up to normalization,
by $\vphantom{\Big|}\sum_{x^n\in T_{p,n}^\delta} p(x^n)\, \psi^{(x^n)}_{A^nB^n}$, where $x^n\coloneqq\qty(x_1,\ldots,x_n)$, $p(x^n) \coloneqq \prod_{i=1}^n p(x_i)$, $\psi^{(x^n)}_{A^nB^n} \coloneqq \bigotimes_{i=1}^n \psi^{(x_i)}_{A_iB_i}$,
and $T_{p,n}^\delta$ is the strongly typical set~\cite[Definition~14.7.2]{MARK}.
Loosely speaking, $T_{p,n}^\delta$ comprises all sequences that contain each symbol $x\in \mathcal{X}$ approximately $N p\qty(x)$ times, with the approximation being controlled by $\delta$. Now, the key observation is that in finite dimension, when $\mathcal{X}$ can be assumed to be finite without loss of generality~\cite{Uhlmann1998}, each of the numbers $N p\qty(x)$ is either $0$ (if $p(x)=0$) or asymptotically large (if $p(x)>0$). In the latter, nontrivial case, we can thus apply the \emph{asymptotic} entanglement dilution protocol for pure states, with associated cost $S\big(\psi_x^A\big)$.

If the alphabet $\mathcal{X}$ is infinite, which will be the case for a generic infinite-dimensional state, then the strongly typical set~\cite[Definition~2]{Ho_2010_unified_typicality} (see also Ref.~\cite{YEUNG-BOOK}) will include sequences in which some of the symbols, namely those with very small probabilities, appear a \emph{non-asymptotic number of times}. We could call this the ``rare'' part of the state $\rho_{AB}^{\otimes n}$. What entanglement dilution protocol do we use for this rare part, given that we cannot rely anymore on the asymptotic protocol?
We will solve this problem by devising a more ``wasteful'' protocol that is to be applied specifically to the rare part of the state. The crux of our argument is to show that the corresponding waste of resources can be controlled and made arbitrarily small asymptotically.

\item
\label{issue:3}
At the same time, the converse part of the argument needed to establish~\eqref{Hayden_et_al_result} also fails in for infinite-dimensional systems. 
One reason for this failure is because in this case the asymptotic continuity inequality~\eqref{asymptotic_continuity_EoF}, which plays a key role in proving~\eqref{Hayden_et_al_result}, trivializes in the limit $d\to\infty$. And this is for a good reason: the entanglement of formation, like most entropic quantities, ceases to be continuous in the infinite-dimensional case. In practice, the well-known fact that quantum entropy of is discontinuous everywhere if the underlying Hilbert space is infinite dimensional~\cite[Sec.~II.D]{Wehrl} already implies that the entanglement of formation is discontinuous at all pure states. This problem has been identified before in Ref.~\cite{nonclassicality}, where it was termed ``asymptotic continuity catastrophe.'' We will solve it by means of the semi-continuity bound recently established in Ref.~\cite{Shirokov2022}.

\item
\label{issue:4}
Finally, another property used in the converse part of the finite-dimensional argument in Ref.~\cite{Hayden-EC} is the monotonicity of the entanglement of formation under LOCC. It was later shown in Ref.~\cite{Gheorghiu2008} that in finite dimension the monotonicity of the entanglement of formation holds even under separable operations. Indeed, combining this result with the argument in Ref.~\cite{Hayden-EC} leads to the conclusion that in finite dimension the entanglement cost under separable operations is given by the same formula as the right-hand side of~\eqref{Hayden_et_al_result}.
It is then natural to wonder whether an analogous formula holds true for infinite-dimensional systems, but this has been challenging to prove due to the lack of analysis of the monotonicity of the entanglement of formation under infinite-dimensional separable operations.
With our rigorous formulation of $\sep_\N$ and $\sep$ in Sec.~\ref{sec:operations}, we will solve this problem by showing that we can use an infinite-dimensional version of the Schur--Horn Theorem~\cite{KAFTAL20103115} to extend the result of Ref.~\cite{Gheorghiu2008}, first to countably separable and then to all separable operations. This will allow us to establish the converse bound for the entanglement cost under separable operations.
\end{enumerate}

\section{Main results} \label{sec:main_results}

In this section, we present our main result, that is, the general equality between the entanglement cost and the regularized entanglement of formation, whose validity we extend from finite- to infinite-dimensional systems. In fact, our theorem below shows something even stronger, i.e., that all the different notions of entanglement cost corresponding to the classes of operations in~\eqref{eq:different_entanglement_costs} are quantitatively equivalent. Our results can be seen as a complete generalization of the finite-dimensional result in Ref.~\cite{Hayden-EC}.

\begin{thm}[(Entanglement cost for infinite-dimensional physical systems)] \label{cost_thm}
Let $\HH_A$, $\HH_B$ be separable (possibly infinite-dimensional) Hilbert spaces. For any quantum state $\rho_{AB}\in \mathcal{D}\qty(\HH_A\otimes \HH_B)$ with finite local quantum entropy
\bb
\min\qty{S\qty(\rho_{A}),\, S\qty(\rho_{B})} < \infty
\label{cost_thm_finite_local_entropy}
\ee
it holds that
\bb
\label{eq:cost_thm}
E_{c,\,\locc_\to}\qty(\rho) = E_{c,\,\locc_\leftarrow}\qty(\rho) =
E_{c,\, \locc_\N}\qty(\rho) = E_{c,\,\locc}\qty(\rho) = E_{c,\,\sep_\N}\qty(\rho) = E_{c,\,\sep}\qty(\rho) = E_f^\infty\qty(\rho)\, ,
\ee
where $E_{c,\, \mathcal{O}}$ is the entanglement cost under a family of operations $\mathcal{O}$ as defined by~\eqref{eq:entanglement_cost} (see also Sec.~\ref{sec:operations}), and $E_f^\infty$ is the regularized entanglement of formation given by~\eqref{regularized_eof}.
\end{thm}

\begin{proof}
The conclusion follows by combining the direct part (Proposition~\ref{prop:achievabllity}) and the converse part (Proposition~\ref{cost_prop_under_sep}).
\end{proof}

\begin{rem}
In the above Theorem~\ref{cost_thm}, we restrict our attention to states with finite local entropy --- i.e., such that~\eqref{cost_thm_finite_local_entropy} holds. Due to~\eqref{eq:min_S_A_B}, for these states, the continuous and the discrete versions of the entanglement of formation (namely, $E_f$ and $\widetilde{E}_f$, respectively) coincide. Since the same consideration holds for their regularized versions, on the rightmost side of~\eqref{eq:cost_thm} we could equivalently replace $E_f^\infty(\rho)$ with $\widetilde{E}_f^\infty(\rho)\coloneqq \lim_{n\to\infty} \frac1n\, \widetilde{E}_f\big(\rho^{\otimes n}\big)$.
\end{rem}

The next two sections are devoted to a complete proof of the above Theorem~\ref{cost_thm}. In particular, Sec.~\ref{sec:achievability} deals with the achievability statement, while Sec.~\ref{sec:converse} contains the converse part.

\section{Direct part: Construction of a one-way LOCC protocol for entanglement dilution} \label{sec:achievability}

In this section, we construct a family of one-way LOCC protocols for entanglement dilution achieving~\eqref{eq:cost_thm}, proving the direct part of Theorem~\ref{cost_thm}. 
After a brief introduction on typicality in Sec.~\ref{subsec:typicality}, we show properties of the entanglement cost in Sec.~\ref{subsec:properties}, which are to be used for our analysis.
Then, we give the proof of the direct part in two main steps: first, in Sec.~\ref{subsec:pure_states}, we deal with the case of pure states, solving the issue~\ref{issue:1} highlighted in Sec.~\ref{subsec:what_goes_wrong}; second, in Sec.~\ref{subsec:general_dilution}, we present our general one-way LOCC protocol for entanglement dilution in the case of mixed states, thus addressing the issue~\ref{issue:2} in Sec.~\ref{subsec:what_goes_wrong} and completing the conceptual developments needed to establish a full achievability proof.

\subsection{Notions of typicality} \label{subsec:typicality}

As in the finite-dimensional case, the crucial notion for constructing an asymptotically optimal entanglement dilution protocol is that of typicality~\cite{CT, YEUNG-BOOK}, which we summarize here. Let $\mathcal{X}=\qty{1,2,\ldots}$ denote a finite or countably infinite set, $X$ a random variable taking on values in $\mathcal{X}$, and $p(x)$ the corresponding probability distribution. A sequence of $n$ independent and identically distributed (IID) copies $X_i$ of $X$ will be denoted by $X^n=\qty(X_1,X_2,\ldots,X_n)$. We write realizations of $X^n$ as $x^n=\qty(x_1,x_2,\ldots,x_n)\in \mathcal{X}^n$; the corresponding probability thus is $p_{n}\qty(x^n)=p\qty(x_1)p\qty(x_2)\cdots p\qty(x_n)$. Define $q\qty(x|x^n)$ as the empirical distribution of symbols $x\in \mathcal{X}$ in the string $x^n$, i.e.,
\bb
q\qty(x|x^n) \coloneqq \frac{N\qty(x|x^n)}{n}\, ,
\ee
where $N\qty(x|x^n)$ is the number of occurrences of $x\in\mathcal{X}$ in the sequence $x^n$. 
We are now ready to define the main notions of typicality employed here, namely, strong typicality, which can be applied (and will be applied) to finite alphabets only, and weak typicality, which is more generally applicable albeit less powerful. For details, we refer the reader to Refs.~\cite{CT, YEUNG-BOOK, WILDE}. A succinct discussion can be found in Ref.~\cite{Ho_2010_unified_typicality}.

If $\mathcal{X}$ is finite, then for any positive integer $n$ and any $\delta>0$, we can define the \deff{strongly typical set} $T_\delta^{X^n}$ as the set of sequences $x^n\in \mathcal{X}^n$ such that, for all $x\in \mathcal{X}$, if $p(x)=0$, then $q(x|x^n) = 0$, and otherwise $\left|p(x) - q(x|x^n)\right| \leq \delta$. The probability that a sequence is strongly typical approaches one asymptotically, i.e.,
\bb
\lim_{n\to\infty} p_n\big(T^{X^n}_\delta\big) = 1 \qquad \forall\ \delta>0\, .
\label{unit_prob_strongly_typical_set}
\ee

An alternative notion that can be used for a countably infinite set is that of weak typicality. If $\mathcal{X}$ is either a finite or countably infinite set, and the random variable $X$ on $\mathcal{X}$ has finite entropy $H(X)\coloneqq-\sum_{x\in\mathcal{X}} p(x)\log_2 p(x)<\infty$, then for any positive integer $n$ and any $\delta>0$, we can define the \deff{weakly typical set} $W_\delta^{X^n}$ as the set of sequences $x^n\in \mathcal{X}^n$ such that
\bb
\left|-\frac1n\, \log_2 p_n(x^n) -H(X)\right|\leq\delta\, .
\ee

Weakly typical sets enjoy a wealth of properties, presented in detail in the aforementioned textbooks. We summarize some of them below:
\begin{enumerate}[(a)]
\item \emph{Asymptotic equipartition property:} 
\bb
2^{-n\left( H(X) + \delta \right)} \leq p_n(x^n) \leq 2^{-n\left( H(X) - \delta \right)} \qquad \forall\ x^n\in W^{X^n}_\delta\, .
\label{AEP_weakly_typical_set}
\ee
\item \emph{Asymptotically unit probability:}
\bb
\lim_{n\to\infty} p_n\big(W^{X^n}_\delta\big) = 1\, ,
\label{unit_prob_weakly_typical_set}
\ee
where $p_n\big(W^{X^n}_\delta\big) = \sum_{x^n\in W^{X^n}_\delta} p_n(x^n)$.
\item \emph{Exponentially small cardinality:}
\bb
\left| W^{X^n}_\delta \right| \leq 2^{n\left(H(X) + \delta\right)} ,\qquad \lim_{n\to\infty} \frac{\left| W^{X^n}_\delta \right|}{2^{n\left(H(X) + \delta\right)}} = 1\, .
\label{exp_smaller_weakly_typical_set}
\ee
\end{enumerate}

\subsection{Properties of the entanglement cost}
\label{subsec:properties}

In this section, we show the properties of the entanglement cost relevant to our analysis.
In particular, after discussing the additivity of the entanglement cost, we show the convexity of the entanglement cost.

First, we show the additivity of entanglement cost as follows.
\begin{lemma}[(Additivity of entanglement cost)] \label{lemma:additivity}
Let $A$ and $B$ be quantum systems represented by separable Hilbert spaces. For any bipartite state $\rho_{AB}\in\mathcal{D}\qty(\HH_A \otimes \HH_B)$, for all classes of operations $\mathcal{O}$ in~\eqref{eq:classes}, and for all $n\in\mathbb{N}_+$, it holds that
\bb
E_{c,\,\mathcal{O}}\qty(\rho)=\frac{1}{n}E_{c,\,\mathcal{O}}\qty(\rho^{\otimes n}).
\ee
\end{lemma}

\begin{proof}
The additivity of $E_{c,\,\mathcal{O}}$ on multiple copies of $\rho$ has been shown in Ref.~\cite[Lemma~1]{Hayden-EC}. Importantly, the proof there works just as well in infinite-dimensional systems. Note that here we are surreptitiously using the fact that the set $\mathcal{O}$ is closed under parallel use of operations (Remark~\ref{closed_under_parallel_composition_rem}).
\end{proof}

An important step in generalizing the finite-dimensional result in Ref.~\cite{Hayden-EC} is to realize that \emph{the entanglement cost is a convex function of the state}. This key fact has been first established in Ref.~\cite{Donald2002}; however, a closer look at the justification given after Definition~17 on p.14 therein reveals that the argument of Ref.~\cite{Donald2002} relies on the expression of the entanglement cost as the regularized entanglement formation, the very result that we are trying to generalize. To avoid circular reasoning, we give here a separate simple proof of the convexity of the entanglement cost, which we consider part of the folklore.

\begin{lemma}[(Convexity of entanglement cost)] \label{lemma:convexity_cost}
Let $A$ and $B$ be quantum systems represented by separable Hilbert spaces $\mathcal{H}_A$ and $\mathcal{H}_B$, respectively, and let $X$ be a random variable with finite alphabet $\XX$ distributed according to some probability distribution $p$. Given a collection of bipartite states $\big(\rho_x\in\mathcal{D}\qty(\mathcal{H})\big)_{x\in \XX}$ with $\mathcal{H}=\mathcal{H}_A\otimes\mathcal{H}_B$ and a class of free operations $\mathcal{O}$ as in~\eqref{eq:classes},
it holds that
\bb
E_{c,\,\mathcal{O}}\left(\sumno_x p(x)\, \rho_x\right) \leq \sumno_x p(x)\, E_{c,\,\mathcal{O}}(\rho_x)\, .
\label{eq:convexity_cost}
\ee
\end{lemma}

\begin{proof}
Up to iterations, we can assume that $X$ is a binary random variable with alphabet $\XX = \{0,1\}$. Also, without loss of generality, we can posit that $p(x)>0$ and $E_{c,\, \mathcal{O}}(\rho_x) < \infty$ for all $x$.
Defining
$\widebar{\rho} = p(0)\, \rho_0 + p(1)\,\rho_1\in\mathcal{D}\qty(\mathcal{H})$,
we need to prove that
\bb
E_{c,\,\mathcal{O}}\big(\widebar{\rho}\big) \leq p(0)\, E_{c,\,\mathcal{O}}(\rho_0) + p(1)\, E_{c,\,\mathcal{O}}(\rho_1)\, .
\label{eq:convexity_cost_binary}
\ee
To this end, it suffices to construct, for all $\eta>0$, an entanglement dilution protocol under $\mathcal{O}$ that prepares the state $\widebar{\rho}$ by consuming $r = p(0)\, E_{c,\,\mathcal{O}}(\rho_0) + p(1)\, E_{c,\,\mathcal{O}}(\rho_1) +\eta$ ebits per output copy. We can assume without loss of generality that $p(0),\, p(1)\in (0,1)$.
We start by expanding
$\widebar{\rho}^{\,\otimes n}\in\mathcal{D}\qty(\mathcal{H}^{\,\otimes n})$,
where $n$ is asymptotically large.
Up to an exponentially small remainder, the typical term will have about $n p(0)$ copies of $\rho_0$ and $n p(1)$ copies of $\rho_1$. More rigorously, for some small $\xi > 0$ such that $0<\xi<\min\{p(0),\,p(1)\}$ --- the precise value will be determined later --- we can define
\bb
\omega_n \coloneqq&\ \frac{1}{p_n\big(T_\xi^{X^n}\big)} \sum_{x^n \in T_\xi^{X^n}} p_n(x^n)\, \bigotimes_{i=1}^n \rho_{x_i} \\
=&\ \frac{1}{p_n\big(T_\xi^{X^n}\big)} \sum_{\substack{k\in \{0,\ldots, n\}:\\ \left|k/n - p(0) \right| \,\leq\, \xi}} \binom{n}{k} \,p(0)^k p(1)^{n-k} \mathcal{S}_n\left( \rho_0^{\otimes k} \otimes \rho_1^{\otimes (n-k)} \right) ,
\label{convexity_cost_proof_eq1}
\ee
where
\bb
\mathcal{S}_n (\cdot) \coloneqq \frac{1}{n!} \sum_{\pi\in S_n} U_\pi^{A^n}\otimes U_\pi^{B^n} \,(\cdot)\, \big(U_\pi^{A^n}\otimes U_\pi^{B^n}\big)^\dag
\label{convexity_cost_proof_eq2}
\ee
is the map that symmetrizes over the $n$ copies of $AB$ --- here, $S_n$ is the symmetric group over $n$ elements,
and $U_\pi^{A^n}$ (respectively $U_\pi^{B^n}$) denotes the unitary operator that permutes the $n$ subsystems of $A^n$ (respectively $B^n$) according to the permutation $\pi\in S_n$.
In going from the first to the second line of~\eqref{convexity_cost_proof_eq1}, we reorganized the sequences $x^n$ by introducing $k \coloneqq N(0|x^n)$, which parametrizes the number of times $0$ appears in the sequence $x^n$.

Strong typicality applied to $X$ shows that
\bb
\left\| \omega_n - \widebar{\rho}^{\,\otimes n} \right\|_1 \tendsn{} 0\, .
\label{convexity_cost_proof_eq3}
\ee
In fact, due to~\eqref{unit_prob_strongly_typical_set}, it holds that
\bb
\left\| \omega_n - \widebar{\rho}^{\,\otimes n} \right\|_1 &\leq \left\| \omega_n - p_n\big(T_\xi^{X^n}\big)\, \omega_n\right\|_1 + \left\| p_n\big(T_\xi^{X^n}\big)\, \omega_n - \widebar{\rho}^{\,\otimes n} \right\|_1  \leq 2 \left(1 - p_n\big(T_\xi^{X^n}\big)\right) \tendsn{} 0\, .
\ee

Now, in light of~\eqref{convexity_cost_proof_eq1} and~\eqref{convexity_cost_proof_eq3}, and taking into account that the symmetrization in~\eqref{convexity_cost_proof_eq2} can be realized with local operations and a source of shared randomness, it is intuitively clear that the entanglement cost of $\widebar{\rho}$ cannot be larger than the rate at which we need to consume entanglement to generate $\rho_0^k \otimes \rho_1^{\otimes (n-k)}$. Since $k\approx n p(0)$, this rate is approximately equal to $p(0)\, E_{c,\,\mathcal{O}}(\rho_0) + p(1)\, E_{c,\,\mathcal{O}}(\rho_1)$, which would conclude the proof.

More rigorously, there exists an operation $\Lambda_n\in \mathcal{O}$ such that
$\Lambda_n$ is a CPTP map from $\mathcal{B}\qty(\mathcal{H}^{\,\otimes \qty(\floor*{n\left(p(0) + \xi\right)}+\floor*{n\left( p(1) + \xi \right)})})$ to  $\mathcal{B}\qty(\mathcal{H}^{\,\otimes n})$ with the property that
\bb
\omega_n = \Lambda_n \left( \rho_0^{\otimes \floor*{n\left(p(0) + \xi\right)}} \otimes \rho_1^{\otimes \floor*{n\left( p(1) + \xi \right)}} \right) .
\label{convexity_cost_proof_eq4}
\ee
To provide a detailed description of $\Lambda_n$, recalling that $\mathcal{H}=\mathcal{H}_A\otimes\mathcal{H}_B$, we write the space for its input system of $\Lambda_n$ as
\bb
\mathcal{H}^{\,\otimes \qty(\floor*{n\left(p(0) + \xi\right)}+\floor*{n\left( p(1) + \xi \right)})}&=\mathcal{H}_{A_1}\otimes\cdots\otimes\mathcal{H}_{A_{\qty(\floor*{n\left(p(0) + \xi\right)}+\floor*{n\left( p(1) + \xi \right)})}} \\
&\qquad\otimes \mathcal{H}_{B_1}\otimes\cdots\otimes\mathcal{H}_{B_{\qty(\floor*{n\left(p(0) + \xi\right)}+\floor*{n\left( p(1) + \xi \right)})}},
\ee
where $\rho_0^{\otimes \floor*{n\left(p(0) + \xi\right)}}$ in~\eqref{convexity_cost_proof_eq4} belongs to
\bb
\label{eq:systems_rho_0}
\rho_0^{\otimes \floor*{n\left(p(0) + \xi\right)}}\in\mathcal{D}\qty(\qty(\mathcal{H}_{A_1}\otimes\mathcal{H}_{B_1})\otimes\cdots\otimes\qty(\mathcal{H}_{A_{\floor*{n\left(p(0) + \xi\right)}}}\otimes\mathcal{H}_{B_{\floor*{n\left(p(0) + \xi\right)}}})),
\ee
and $\rho_1^{\otimes \floor*{n\left(p(1) + \xi\right)}}$ belongs to the remaining part of the systems.
The operation $\Lambda_n$ consists of the following steps:
\begin{enumerate}[(i)]
\item Alice and Bob draw a uniformly distributed permutation $\pi\in S_n$ and an integer-valued random variable $K\in \{0,\ldots, n\}$ that is distributed according to the ``curtailed'' binomial distribution
\bb
\pr \{ K = k\} = \left\{ \begin{array}{ll} \frac{1}{p_n\left(T_\xi^{X^n}\right)}\, \binom{n}{k} p(0)^k p(1)^{n-k} & \quad\text{if $\left|\frac{k}{n} - p(0)\right| \leq \xi$,} \\[1.5ex] 0 & \quad\text{otherwise}; \end{array} \right.
\ee

\item noting that $k \leq \floor*{n (p(0) + \xi)}$ and $n-k \leq \floor*{n(p(1) + \xi)}$, Alice and Bob discard $\floor*{n (p(0) + \xi)} - k$ pairs of subsystems $A_i B_i$ from those appearing in~\eqref{eq:systems_rho_0}, and also $\floor*{n(p(1) + \xi)} - n+k$ pairs of subsystems $A_iB_i$ from the remaining ones pertaining to the state $\rho_1^{\otimes \floor*{n\left(p(1) + \xi\right)}}$; note that taking partial traces over local subsystems is an operation in $\mathcal{O}$; after this, they are left with exactly $n$ pairs of subsystems, with Alice's local system written as $A^n$ and Bob's as $B^n$;

\item both Alice and Bob then permute their local systems $A^n$ and $B^n$ by applying the unitaries $U_\pi^{A^n}$ and $U_\pi^{B^n}$;

\item finally, they discard the classical system on which they have stored the values of $\pi$ and $k$, effectively ``forgetting'' them; since discarding system is in $\mathcal{O}$, so is this operation; observe that forgetting $\pi$ after step~(iii) results in the effective application of the symmetrization operator $\mathcal{S}_n$ defined by~\eqref{convexity_cost_proof_eq2}.

\end{enumerate}

By construction, $\Lambda_n$ makes use of local operations and shared randomness only. Remembering also the discussion in Remark~\ref{common_randomness_rem}, this proves that $\Lambda_n\in \mathcal{O}$. Also, upon acting on the input state $\rho_0^{\otimes \floor*{n\left(p(0) + \xi\right)}} \otimes \rho_1^{\otimes \floor*{n\left( p(1) + \xi \right)}}$, the channel $\Lambda_n$ clearly outputs the state $\omega_n$ defined by~\eqref{convexity_cost_proof_eq1}; it thus satisfies~\eqref{convexity_cost_proof_eq4}.

The above argument already suffices to conclude the proof. In fact, set $r_x \coloneqq E_{c,\,\mathcal{O}}(\rho_x) + \zeta$, where $x=0,1$, and $\zeta>0$ is a small parameter to be fixed later. Since $r_x$ is an achievable rate of entanglement dilution for $\rho_x$, we can consider sequences of operations $\EE_{x,m} \in \mathcal{O}$ such that
\bb
\tau_{x,m} \coloneqq \EE_{x,m} \Big(\Phi_2^{\otimes \floor*{r_x m}}\Big)\, ,\qquad \left\| \tau_{x,m} - \rho_x^{\otimes m} \right\|_1 \tends{}{m\to\infty} 0\, .
\label{convexity_cost_proof_eq5}
\ee
We then see that the state
\bb
\widebar{\omega}_n \coloneqq&\ \left(\Lambda_n \circ \big( \EE_{0,\floor*{n\left(p(0) + \xi\right)}} \otimes \EE_{1,\floor*{n\left(p(1) + \xi\right)}} \big)\right) \Big( \Phi_2^{\otimes \sum_x \floor*{r_x \floor*{n\left(p(x) + \xi\right)}}} \Big) \\
=&\ \left(\Lambda_n \circ \big( \EE_{0,\floor*{n\left(p(0) + \xi\right)}} \otimes \EE_{1,\floor*{n\left(p(1) + \xi\right)}} \big)\right) \Big( \Phi_2^{\otimes \floor*{r_0 \floor*{n\left(p(0) + \xi\right)}}} \otimes \Phi_2^{\otimes \floor*{r_1 \floor*{n\left(p(1) + \xi\right)}}} \Big) \\
=&\ \,\Lambda_n \left( \tau_{0,\, \floor*{n\left(p(0) + \xi\right)}} \otimes \tau_{1,\, \floor*{n\left(p(1) + \xi\right)}} \right) ,
\ee
satisfies that
\bb
\left\| \widebar{\omega}_n - \widebar{\rho}^{\,\otimes n} \right\|_1 &\leq \left\| \widebar{\omega}_n - \omega_n \right\|_1 + \left\| \omega_n - \widebar{\rho}^{\,\otimes n} \right\|_1 \\
&\leq \left\| \tau_{0,\, \floor*{n\left(p(0) + \xi\right)}} \otimes \tau_{1,\, \floor*{n\left(p(1) + \xi\right)}} - \rho_0^{\otimes \floor*{n\left(p(0) + \xi\right)}} \otimes \rho_1^{\otimes \floor*{n\left( p(1) + \xi \right)}} \right\|_1 + \left\| \omega_n - \widebar{\rho}^{\,\otimes n} \right\|_1 \\
&\leq \left\| \tau_{0,\, \floor*{n\left(p(0) + \xi\right)}} - \rho_0^{\otimes \floor*{n\left(p(0) + \xi\right)}} \right\|_1 + \left\| \tau_{1,\, \floor*{n\left(p(1) + \xi\right)}} - \rho_1^{\otimes \floor*{n\left( p(1) + \xi \right)}} \right\|_1 + \left\| \omega_n - \widebar{\rho}^{\,\otimes n} \right\|_1 \\
&\tendsn{} 0\, .
\label{convexity_cost_proof_eq6}
\ee
Here, the first and third lines are simple applications of the triangle inequality, while in the second line, we used the data processing inequality for the trace norm to eliminate the channel $\Lambda_n$. To take the limit on the last line, instead, we leveraged~\eqref{convexity_cost_proof_eq3} and~\eqref{convexity_cost_proof_eq5}.

Since $\mathcal{O}$ is closed under parallel use by Remark~\ref{closed_under_parallel_composition_rem}, it holds that $\EE_{0,m} \otimes \EE_{1,m^\prime} \in \mathcal{O}$ for all $m$ and $m^\prime$.
Thanks to the facts that we also have $\Lambda_n \in \mathcal{O}$ and that $\mathcal{O}$ is closed under composition by Remark~\ref{common_randomness_rem}, we deduce that $\Lambda_n \circ \big( \EE_{0,m} \otimes \EE_{1,m^\prime} \big)\in \mathcal{O}$.
Then,~\eqref{convexity_cost_proof_eq6} and the definition~\eqref{eq:entanglement_cost} of entanglement cost show that
\bb
E_{c,\,\mathcal{O}}\big(\widebar{\rho}\big) &\leq \limsup_{n\to\infty} \frac1n \sum_x \floor*{r_x \floor*{n\left(p(x) + \xi\right)}} \\
&= \sum_x r_x \left(p(x) + \xi\right) \\
&= \sum_x \left(p(x) + \xi\right) \left(E_{c,\,\mathcal{O}}(\rho_x) + \zeta\right) \\
&= \sum_x p(x)\, E_{c,\,\mathcal{O}}(\rho_x) + \zeta + \xi \sum_x E_{c,\,\mathcal{O}}(\rho_x) + 2 \zeta \xi\, .
\ee
Since we assumed that $E_{c,\,\mathcal{O}}(\rho_x) < \infty$ without loss of generality, we can now take $\xi$ and $\zeta$ small enough so that the right-hand side of the above expression is no larger than $\eta$, thereby concluding the proof.
\end{proof}

\subsection{Achievability of entanglement cost for infinite-dimensional pure states} \label{subsec:pure_states}

We are now ready to address the first of the issues listed in 
Sec.~\ref{subsec:what_goes_wrong}. To do this, we have to design an entanglement dilution protocol that can prepare any (finite- or infinite-dimensional) pure state $\psi_{AB}$ with finite local quantum entropy $S(\psi_A)<\infty$ consuming about $S(\psi_A)$ ebits per output copy. The main idea of this protocol is the same as in the finite-dimensional case: we simply prepare the locally typical part of the state. The only difference is that we have to use the notion of weak typicality for countably infinite alphabets instead of finite ones. A similar approach has appeared before in Ref.~\cite[proof of Proposition~3]{Wilde2018}.

\begin{lemma}[(Achievability of entanglement cost for pure states)] \label{lemma:approximate_bipartite}
Let $A$ and $B$ be quantum systems represented by separable Hilbert spaces. Let $\psi_{AB} = \ketbra{\psi}_{AB}$ be a bipartite pure state on $AB$ with finite entanglement entropy $S(\psi_A) = S(\psi_B) < \infty$. Then, the entanglement entropy is an achievable rate for entanglement dilution to $\psi_{AB}$ under one-way LOCC, i.e.
\bb
E_{c,\, \locc_\to}(\psi_{AB}) \leq S(\psi_A) = S(\psi_B)\, .
\ee
\end{lemma}

\begin{proof}
Let the Schmidt decomposition of $\ket{\psi}_{AB}$ be\footnote{The Schmidt decomposition relies exclusively on the singular value decomposition for compact operators applied to the matrix of coefficients of $\ket{\psi}_{AB}$ with respect to a product of orthonormal bases. This works in all separable Hilbert spaces, finite- or infinite-dimensional.}
\bb
\ket{\psi}_{AB}=\sum_{x=0}^\infty \sqrt{p\qty(x)}\, \ket{e_x}_A\otimes\ket{f_x}_B\, ,
\ee
where $\left\{\ket{e_x}_A\right\}_x$ and $\left\{\ket{f_x}_B\right\}_x$ are local orthonormal bases, and $p$ is a probability distribution. Let $X$ denote the associated random variable. Observe that
\bb
H(X) = -\sum_{x=0}^\infty p(x) \log_2 p(x) = S(\psi_A) < \infty\, .
\ee
For an arbitrary $\delta>0$ and any positive integer $n$, set
\bb
\ket{\psi^{(n)}}_{A^nB^n} \coloneqq \frac{1}{\sqrt{p_n\big(W_\delta^{X^n}\big)}} \sum_{x^n \in W_\delta^{X^n}} \sqrt{p_n(x^n)}\, \ket{e_{x^n}}_{A^n} \otimes\ket{f_{x^n}}_{B^n}\, ,
\ee
where for a sequence $x^n = x_1\ldots x_n$, we set
\bb
\ket{e_{x^n}}_{A^n} \coloneqq \bigotimes_{i=1}^n \ket{e_{x_i}}_{A_i}\, ,\qquad \ket{f_{x^n}}_{B^n} \coloneqq \bigotimes_{i=1}^n \ket{f_{x_i}}_{B_i}\, .
\ee
By~\eqref{unit_prob_weakly_typical_set}, we have that $\braket{\psi^{(n)}}{\psi^{\otimes n}} = \sqrt{p_n\big(W_\delta^{X^n}\big)} \tendsn{} 1$, implying that also
\bb
\frac12 \left\|\psi^{(n)} - \psi^{\otimes n}\right\|_1 = \sqrt{1- \left|\braket{\psi^{(n)}}{\psi^{\otimes n}}\right|^2} = \sqrt{1- p_n\big(W_\delta^{X^n}\big)} \tendsn{} 0\, .
\label{approximation_pure}
\ee

Now, there exists a simple one-way LOCC protocol that prepares any pure state by consuming $\ceil{\log_2 r}$ ebits, where $r$ is the local rank, i.e., the rank of the reduced state. This protocol, as described in the finite-dimensional case by Ref.~\cite{Bennett-distillation}, consists of teleporting one share of the pure state from Alice to Bob (this is the direction of classical communication in $\locc_\to$). In our case, the local rank of $\psi^{(n)}$ satisfies that
\bb
\rk \left(\psi^{(n)}_{A^n}\right) = \left| W_\delta^{X^n} \right| \leq 2^{n\left(H(X) + \delta\right)} = 2^{n\left(S(\psi_A) + \delta\right)} ,
\ee
where we made use of~\eqref{exp_smaller_weakly_typical_set}. Hence, we can prepare $\psi^{(n)}$ by consuming at most $\ceil{n\left(S(\psi_A) + \delta\right)}$ ebits. Due to~\eqref{approximation_pure}, we therefore see that $S(\psi_A) + \delta$ is an achievable rate of entanglement dilution, i.e., $E_{c,\,\locc_\to}(\psi_{AB})\leq S(\psi_A) + \delta$. Since $\delta>0$ was arbitrary, this shows that $E_{c,\,\locc_\to}(\psi_{AB})\leq S(\psi_A)$, concluding the proof.
\end{proof}

Before moving on, it is instructive to present a simple corollary of the analysis of entanglement dilution for pure states in combination with the convexity of the entanglement cost. This corollary provides a bound on the entanglement cost for finite-rank infinite-dimensional mixed states, but it still does not suffice to fully characterize the entanglement cost of general infinite-dimensional mixed states.

\begin{cor}[(Achievable rate of entanglement dilution for finite-rank mixed states)] \label{cor:cost_finite_rank}
Let $A$ and $B$ be quantum systems represented by separable Hilbert spaces. For some integer $N\in \N$, let 
\bb
\rho_{AB} = \sum_{x=0}^N p(x)\, \psi_{x}^{AB}
\ee
be a finite-rank bipartite state. Here, $p$ is a probability distribution over $\{0,\ldots, N\}$, and for each $x$ $\psi_x^{AB} = \ketbra{\psi_x}_{AB}$ is a pure state. Then
\bb
E_{c,\,\locc_\to}(\rho_{AB}) \leq \sum_{x=0}^N p(x)\, S\big(\psi_{x}^{A}\big)\, .
\ee
\end{cor}

\begin{proof}
The conclusion follows by combining Lemmas~\ref{lemma:convexity_cost} and~\ref{lemma:approximate_bipartite}.
\end{proof}

\subsection{Achievability of entanglement cost for infinite-dimensional mixed states} 
 \label{subsec:general_dilution}

We now show how to construct a protocol that performs entanglement dilution on an arbitrary mixed state with finite local quantum entropy at a rate equal to the regularized entanglement of formation. We will also see later in Sec.~\ref{sec:converse} that such a rate is, in fact, optimal. To arrive at the result, we have to solve issue~(2) listed in Sec.~\ref{subsec:what_goes_wrong}.
As we have seen with Corollary~\ref{cor:cost_finite_rank}, a mere combination of the result on pure states with the convexity of entanglement cost is insufficient for addressing this issue.
One of our main original contributions is precisely an idea to address this issue, which is based on dividing the state into a ``common'' and a ``rare'' part, applying the standard entanglement dilution protocol to the former, and a different ``wasteful'' protocol to the latter. The standard protocol is based on the subroutine presented in Sec.~\ref{subsec:pure_states}, while the wasteful protocol will be explained below. See Remark~\ref{explicit_protocol_rem} below for more details about how these sub-protocols constitute the overall protocol for the entanglement dilution of infinite-dimensional mixed states.

\begin{prop}[(Achievability of entanglement cost for mixed states)] \label{prop:achievabllity}
Let $A$ and $B$ be quantum systems represented by separable Hilbert spaces. For any bipartite state $\rho_{AB}\in\mathcal{D}\qty(\HH_A \otimes \HH_B)$ with finite local quantum entropy $\min\qty{S\qty(\rho_{A}),S\qty(\rho_{B})}<\infty$ and for all classes of operations $\mathcal{O}$ in~\eqref{eq:classes}, it holds that
\bb
E_{c,\,\mathcal{O}}\qty(\rho)\leq E_f^\infty\qty(\rho)\, ,
\label{regularized_EoF_achievable}
\ee
where $E_f^\infty$ is the regularized entanglement of formation given by~\eqref{regularized_eof}.
\end{prop} 

\begin{proof}
The first observation is that, precisely as in the finite-dimensional case, it suffices to prove that 
\bb
E_{c,\,\mathcal{O}}\qty(\rho)\leq E_f\qty(\rho)
\label{EoF_achievable}
\ee
(without regularization), due to the additivity of the entanglement cost shown in Lemma~\ref{lemma:additivity}. 

We now set out to prove~\eqref{EoF_achievable}. To this end, start by observing that, due to the finiteness of the local entropy, it follows from~\eqref{eq:min_S_A_B} that $E_f(\rho) = \widetilde{E}_f(\rho)$; thus, it suffices to consider discrete decompositions of $\rho$ into pure states. Let
\bb
\rho = \rho_{AB} = \sum_{x=0}^\infty p(x)\, \psi_x^{AB}\, .
\ee
be an arbitrary such decomposition. 
Without loss of generality, up to exchanging the roles of $A$ and $B$, here we assume that $S(\rho_A) <\infty$. 
For an arbitrary integer $N\in \N_+$, set 
\bb
\delta_N \coloneqq \sum_{x=N+1}^\infty p(x)\, ,
\ee
and construct the states
\bb
\rho_N = \rho_N^{AB} \coloneqq \frac{1}{1 - \delta_N} \sum_{x=0}^N p(x)\, \psi_x^{AB}\, , \qquad \omega_N = \omega_N^{AB} \coloneqq \frac{1}{\delta_N} \sum_{x=N+1}^\infty p(x)\, \psi_x^{AB}\, ,
\ee
so that 
\bb
\rho = \left(1 - \delta_N\right) \rho_N + \delta_N \omega_N\, .
\label{achievabllity_proof_eq2}
\ee
Here, $\rho_N$ can be thought of as the ``common'' part of $\rho$, while $\omega_N$ represents its ``rare'' part. Since $\sum_{x=0}^\infty p(x) = 1$ converges, we have that 
\bb
\delta_N \tends{}{N\to\infty} 0\, , 
\label{achievabllity_proof_eq2.5}
\ee
and thus
\bb
\left\| \rho - \rho_N \right\|_1 = \delta_N \left\|\rho_N - \omega_N\right\|_1 \leq 2 \delta_N \tends{}{N\to\infty} 0\, .
\ee
Naturally, the reduced states on the system $A$ also satisfy that 
\bb
\left\| \rho_A - \rho_N^A \right\|_1 \leq \left\| \rho - \rho_N \right\|_1 \tends{}{N\to\infty} 0\, ,
\ee
and by the lower semi-continuity of the quantum entropy\footnote{
Indeed, by applying the dominated convergence theorem for entropy, specifically, Ref.~\cite[Theorem~A3]{10.1063/1.1666274} with the substitution $n=N$, $A_n=\rho_N^A$, $A=\rho_A$, and $B=\rho_A$ from its notation to ours, one can obtain a stronger relation than~\eqref{achievabllity_proof_eq3}: $\lim_{N\to\infty}S\qty(\rho_N^A)=S\qty(\rho_A)$.
This application is justified by $S\qty(\rho_A)<\infty$, $\left\| \rho - \rho_N \right\|_1\tends{}{N\to\infty}0$, and $\rho_N^A=\qty(\rho_A-\delta_N\omega_N^A)/\qty(1-\delta_N)\leq \rho_A$, which together satisfy the assumptions of Ref.~\cite[Theorem~A3]{10.1063/1.1666274}. We thank an anonymous reviewer for bringing this refinement to our attention. Nonetheless, the weaker bound given in~\eqref{achievabllity_proof_eq3} suffices for our purposes.
}
\bb
\liminf_{N\to\infty} S\big(\rho_N^A\big) \geq S(\rho_A)\, .
\label{achievabllity_proof_eq3}
\ee
But by the concavity of the quantum entropy, it also holds that
\bb
S(\rho_A) \geq \left(1 - \delta_N\right) S\big(\rho_N^A\big) + \delta_N\, S\big(\omega_N^A\big)\, ,
\ee
entailing that 
\bb
S(\rho_A) &\geq \limsup_{N\to\infty} \left(\left(1 - \delta_N\right) S\big(\rho_N^A\big) + \delta_N\, S\big(\omega_N^A\big)\right) \\
&\geq \liminf_{N\to\infty} \left(1 - \delta_N\right) S\big(\rho_N^A\big) + \limsup_{N\to\infty} \delta_N\, S\big(\omega_N^A\big) \\
&= \liminf_{N\to\infty} S\big(\rho_N^A\big) + \limsup_{N\to\infty} \delta_N\, S\big(\omega_N^A\big)\, .
\label{achievabllity_proof_eq4}
\ee
Putting~\eqref{achievabllity_proof_eq3} and~\eqref{achievabllity_proof_eq4} together shows that
\bb
\delta_N S\big(\omega_N^A\big) \tends{}{N\to\infty} 0\, .
\label{achievabllity_proof_eq5}
\ee
This key fact will be used later for bounding the cost of the ``wasteful'' part of our protocol.

We now come to the analysis of the entanglement cost of $\rho$. Applying first Lemma~\ref{lemma:convexity_cost} to $\rho$ decomposed as in~\eqref{achievabllity_proof_eq2}, and then Corollary~\ref{cor:cost_finite_rank} to $\rho_N$, we deduce that
\bb
E_{c,\,\mathcal{O}}(\rho) &\leq \left(1 - \delta_N\right) E_{c,\,\mathcal{O}}(\rho_N) + \delta_N E_{c,\,\mathcal{O}}(\omega_N) \\
&= \left(1 - \delta_N\right) \sum_{x=0}^N p(x)\, S\big(\psi_x^A\big) + \delta_N E_{c,\,\mathcal{O}}(\omega_N)\, .
\label{achievabllity_proof_eq6}
\ee
To proceed we need to find a suitable upper bound on $E_{c,\,\mathcal{O}}(\omega_N)$. Here is where our ``wasteful'' protocol comes into play. Consider a purification $\phi_N^{ABE}$ of $\omega_N = \omega_N^{AB}$. Since tracing away the system $E$ is a local operation, we can imagine to upper bound the cost of $\omega_N$ by instead aiming to prepare $\phi_N^{A:BE}$ (the purifying system $E$ is assigned to Bob) and then tracing out $E$. More formally, applying~\eqref{data_processing_cost} to the local operation $\Tr_E \in \mathcal{O}_{A:BE\to A:B}$, we see that
\bb
E_{c,\,\mathcal{O}}(\omega_N) \leq E_{c,\,\mathcal{O}}\Big(\phi_N^{A:BE}\Big) \leq S\big(\phi_N^A\big) = S\big(\omega_N^A\big)\, ,
\ee
where the second inequality descends from Lemma~\ref{lemma:approximate_bipartite}. Plugging this into~\eqref{achievabllity_proof_eq6} yields
\bb
E_{c,\,\mathcal{O}}(\rho) &\leq \left(1 - \delta_N\right) \sum_{x=0}^N p(x)\, S\big(\psi_x^A\big) + \delta_N S\big(\omega_N^A\big)\, .
\ee
We can now take the limit $N\to\infty$. Due to~\eqref{achievabllity_proof_eq2.5} and~\eqref{achievabllity_proof_eq5}, we obtain that
\bb
E_{c,\,\mathcal{O}}(\rho) &\leq \sum_{x=0}^\infty p(x)\, S\big(\psi_x^A\big)\, .
\ee
Since the decomposition of $\rho$ into the pure states is arbitrary, we can now optimize over it. Due to~\eqref{eq:min_S_A_B}, doing so yields precisely~\eqref{EoF_achievable} and concludes the proof.
\end{proof}

\begin{rem} \label{explicit_protocol_rem}
It is instructive to pause for a second and consider what entanglement dilution protocol we have implicitly described in the proof of the above Proposition~\ref{prop:achievabllity}. Essentially, in~\eqref{achievabllity_proof_eq2}, we divided the state $\rho$ into a ``common'' part $\rho_N$ and a ``rare'' part $\omega_N$. Expanding $\rho^{\otimes n}$ therefore produces a sum of strings of tensor products of $\rho_N$ and $\omega_N$. Each string contains a majority of states $\rho_N$, and a small fraction of states $\omega_N$. The procedure we have implicitly outlined involves preparing these two states separately: $\rho_N$, by using its pure state decomposition, and $\omega_N$, by first obtaining its purification and then tracing away the purifying system. The crux of the analysis is to show that this suboptimal, ``wasteful'' protocol, however, does not impact the rate significantly because relatively few copies of $\omega_N$ need to be prepared in the asymptotic limit. At the same time, since the decomposition of $\rho_N$ itself has only finitely many pure states, we have successfully bypassed issue~(2) listed in Sec.~\ref{subsec:what_goes_wrong} on the use of the strong typicality.
\end{rem}

\section{Converse part: Optimality among all entanglement dilution protocols under separable operations} \label{sec:converse}

In this section, we show that no protocol for entanglement dilution, even under separable operations, can achieve a rate that is lower than the regularized entanglement of formation in~\eqref{regularized_eof}. In this way, we will establish the converse part of Theorem~\ref{cost_thm} and, therefore, conclude the proof of our main result.
The main original contribution of this section is an argument for the converse part based on two results. One resolves issue~(3) presented in Sec.~\ref{subsec:what_goes_wrong} by identifying an appropriate replacement for asymptotic continuity in the infinite-dimensional case, i.e., the semi-continuity bound for the entanglement of formation recently introduced in Ref.~\cite{Shirokov2022}. The other is the resolution of issue~(4) in Sec.~\ref{subsec:what_goes_wrong}, i.e., the generalization of the proof of the monotonicity of the entanglement of formation to infinite-dimensional separable operations. We present the former in Sec.~\ref{subsec:one_sided_continuity_EoF} and the latter in Sec.~\ref{subsec:monotonicity_EoF_SEP}. Using these results, we finally prove the converse part in Sec.~\ref{subsec:proof_of_converse}.

\subsection{Semi-continuity of the entanglement of formation} \label{subsec:one_sided_continuity_EoF}

To discuss the continuity of entropic quantities defined for infinite-dimensional quantum systems, a conventional technique is to impose an energy constraint on the states. To this end, one needs to introduce a Hamiltonian, represented by a (possibly unbounded) self-adjoint operator $H$ on the underlying Hilbert space $\mathcal{H}$. For physical reasons, and up to redefining the ground state energy, we will assume that $H$ is positive semidefinite, meaning that $\bra{\psi} H \ket{\psi} \geq 0$ for all $\ket{\psi}\in \dom\qty(H)$, where $\dom(X)$ denotes the domain of a self-adjoint operator $X$. The average energy with respect to $H$ of a state $\rho$ on $\mathcal{H}$ with spectral decomposition $\rho = \sum_{x=0}^\infty p_x \ketbra{\psi_x}$ can then be defined as
\bb
\Tr \rho H \coloneqq \left\{ \begin{array}{ll} \sum_{x=0}^\infty p_x \left\| H^{1/2} \ket{\psi_x} \right\|^2 & \quad \text{if $\ket{\psi_x}\in \dom\qty(H^{1/2})$ whenever $p_x>0$,} \\[1.5ex] \infty &\quad \text{otherwise.} \end{array} \right.
\ee
Note that the above infinite sum contains only non-negative terms and is therefore always well defined, although possibly infinite. Also, a subtle point that is important to take into consideration is that $\Tr \rho H$ may be well defined and finite even if $\supp(\rho)\not\subseteq \dom\qty(H)$ (and thus even if $\rho H$ is not a trace class operator itself).\footnote{Consider, for example, the case where $\rho = \ketbra{\psi}$ with $\ket{\psi} \propto \sum_{n=1}^\infty \frac{1}{n^{3/2}} \ket{n}$, while $H = \sum_{n=1}^\infty n \ketbra{n}$.} 

With that in mind, we can consider a state $\rho$ whose average energy with respect to a Hamiltonian $H$ is bounded by some constant $E\geq 0$, in the sense that
\bb \label{eq:energy_constraint}
\Tr\qty[\rho H]\leq E.
\ee
Intuitively, for such an energy constraint to work as expected, we also need to avoid pathological Hamiltonians $H$. A mathematically convenient condition to avoid such a pathological behavior is the so-called Gibbs hypothesis, which guarantees that $H$ is associated with a well-defined thermal state at all temperatures.

\begin{Def}[(Gibbs hypothesis)] \label{Gibbs_hypothesis_def}
A self-adjoint (possibly unbounded) operator $H$ on a separable Hilbert space $\mathcal{H}$ is called a \deff{grounded Hamiltonian} if the minimum of its spectrum is $0$. We say that $H$ satisfies the \deff{Gibbs hypothesis} if
\bb
\Tr\qty[e^{-\beta H}] < \infty \qquad \forall\ \beta > 0\, ,
\label{Gibbs_hypothesis}
\ee
i.e., if $e^{-\beta H}$ is a trace-class operator for all $\beta>0$.
\end{Def}

For a grounded Hamiltonian that satisfies the Gibbs hypothesis, the Gibbs state
\bb
\gamma_{H,\,\beta} \coloneqq \frac{e^{-\beta H}}{\Tr e^{-\beta H}}
\ee
is well defined for all inverse temperatures $\beta>0$; hence, it also has well-defined energy
\bb
\Tr\qty[\gamma_{H,\,\beta}H]= E\, . 
\ee
In fact, it is most common to use the above equation to re-parametrize $\beta$ as a function of $E$. In this way, we obtain the function $\beta_H:[0,\infty) \to (0,\infty]$, with $\beta_H(0) = \infty$ and $\lim_{E\to\infty} \beta_H(E) = 0$. A key role in the theory is played by the function
\bb
F_H(E) \coloneqq S\big(\gamma_{H,\,\beta_H(E)}\big)\, ,
\label{FH}
\ee
which records the entropy of the Gibbs state of a given energy. If $H$ is grounded and satisfies the Gibbs hypothesis, $F_H$ is well-defined for all $E\geq 0$ and enjoys a wealth of nice properties, some of which have been proved in Ref.~\cite[Proposition~1]{Shirokov-1}. We are only going to need the following.

\begin{lemma}[{(Properties of the function $F_H$~\cite[Proposition~1]{Shirokov-1})}] \label{F_H_o_of_E_lemma}
Let $H$ be a grounded Hamiltonian that satisfies the Gibbs hypothesis. Then, the function $F_H:[0,\infty) \to [0,\infty)$ defined by~\eqref{FH} is finite on the whole half-line $[0,\infty)$, and moreover
\bb
F_H(E) = o(E)\qquad (E\to\infty)\, ,
\label{F_H_o_of_E}
\ee
i.e., $\lim_{E\to\infty} \frac{F_H(E)}{E} = 0$.
\end{lemma}

\begin{rem}
    To obtain Lemma~\ref{F_H_o_of_E_lemma} from Ref.~\cite[Proposition~1]{Shirokov-1}, one needs to use the relations proved in point~(ii) there, and observe that, in the notation of Ref.~\cite{Shirokov-1}, the Gibbs hypothesis is equivalent to $g(H)=0$.
\end{rem}

We will now establish a semi-continuity bound on the entanglement of formation. This is needed to relate the entanglement of formation of the target state and that of the output state of the entanglement dilution protocol, namely, $\rho^{\otimes n}$ and $\rho_n \coloneqq \EE\big(\Phi^{\otimes \floor*{rn}}\big)$, respectively (cf.~\eqref{eq:entanglement_cost}). We are free to study whatever target states we wish, so it makes perfect sense to impose a mild condition on $\rho$ such as~\eqref{eq:min_S_A_B}. We cannot impose an energy constraint on $\rho_n$: in fact, for the converse part, we need to analyze an arbitrary entanglement dilution protocol, which may inject an arbitrarily high amount of energy into the system to output $\rho_n$. Due to this asymmetry on the constraints between $\rho^{\otimes n}$ and $\rho_n$, bounds with energy constraints on both $\rho$ and $\rho_n$ are inappropriate for our analysis.
By contrast, the characteristic property of the following bound is that it is a \emph{one-side} continuity bound featuring an energy constraint on $\rho$ only: as such, it is suitable for our analysis.

\begin{lemma}[{(Semi-continuity bound for the entanglement of formation~\cite[Proposition~4.B]{Shirokov2022})}] \label{alsc_formation_lemma}
Let $H_A$ be a grounded Hamiltonian on a quantum system $A$ that satisfies the Gibbs hypothesis (Definition~\ref{Gibbs_hypothesis_def}). Then, 
given an arbitrary bipartite state $\rho_{AB}\in \mathcal{D}\qty(\mathcal{H}_A\otimes \mathcal{H}_B)$ such that $E = \Tr\qty[ \rho_{A} H_A] < \infty$, with $\rho_A \coloneqq \Tr_B\qty[\rho_{AB}]$ being the reduced state on $A$, for any state $\tilde{\rho}_{AB}$ we have that
\bb
E_f\qty(\rho) - E_f\qty(\tilde{\rho}) \leq \epsilon^\prime F_H\left(\frac{E}{\epsilon^\prime}\right) + g\qty(\epsilon^\prime),
\ee
where $E_f$ is the entanglement of formation defined by~\eqref{eof}, $F_H$ is the function given by~\eqref{FH} (which is well defined because of the Gibbs hypothesis), $\epsilon^\prime \coloneqq \sqrt{\e\qty(2-\e)}$ with $\e\coloneqq \frac12 \left\|\tilde\rho - \rho\right\|_1$, and $g\qty(x)\coloneqq \qty(x+1)\log_2\qty(x+1) - x \log_2 x$.
\end{lemma}

We also have the following lemma to convert the constraint in terms of energy into that in terms of entropy.
Since our condition on $\rho_{AB}$ is given in terms of the entropy as~\eqref{eq:min_S_A_B}, the lemma below will be useful for converting our condition into the energy-constraint condition for the above asymptotic semi-continuity bound.

\begin{lemma}[{(Hamiltonian associated with a state~\cite[Proposition~4]{Shirokov-1})}] \label{associated_Hamiltonian_lemma}
For an arbitrary state $\rho\in \mathcal{D}\qty(\mathcal{H})$, the following are equivalent:
\begin{enumerate}[(a)]
\item there exists a grounded Hamiltonian $H$ that satisfies the Gibbs hypothesis (Definition~\ref{Gibbs_hypothesis_def}) and such that $\Tr\qty[\rho H]<\infty$;
\item $S\qty(\rho)<\infty$.
\end{enumerate}
\end{lemma}

\begin{proof}
     We give a simplified and fully self-contained proof in Appendix~\ref{sec:proof_associated_hammltonian_lemma}.
\end{proof}

\subsection{Monotonicity of the entanglement of formation under separable operations} \label{subsec:monotonicity_EoF_SEP}

Throughout this section, we extend the result of Ref.~\cite[Sec.~III, item~v)]{Gheorghiu2008} on the monotonicity of the entanglement of formation under finite-dimensional separable operations to infinite-dimensional separable operations. 

In the finite-dimensional case, the possible ensembles of bipartite pure states produced when separable operations act on a single bipartite pure state are characterized in Ref.~\cite{Gheorghiu2008}. We start by generalizing this finite-dimensional result to the present infinite-dimensional setting. The analysis in Ref.~\cite{Gheorghiu2008} relies on a function defined as the sum of $N$ smallest eigenvalues of a positive semidefinite operator. In particular, for a $D$-dimensional positive semidefinite operator $M$ with eigenvalues $\lambda_0 \geq \lambda_1 \geq \cdots \geq \lambda_{D-1} \geq 0$ and spectral decomposition $M=\sum_{n=0}^{D-1}\lambda_n\ketbra{n}$, Ref.~\cite{Gheorghiu2008} defines
\bb
\label{eq:chi_finite}
    \chi_N\qty(M) &\coloneqq \sum_{n=D-N}^{D-1}\lambda_{n}\, ,
\ee
with the convention that $\chi_0\qty(M)\coloneqq 0$. The function $\chi_N$ played a central role in the finite-dimensional analysis of Ref.~\cite{Gheorghiu2008}, but, problematically, $\chi_N$ by itself is not well defined in the infinite-dimensional case, for the simple reason that there might be no ``smallest eigenvalue''. Instead, for any positive semidefinite trace-class operator with spectral decomposition
\bb
T=\sum_{n=0}^{\infty}\lambda_n\ketbra{n}\, ,\quad \lambda_0\geq \lambda_1\geq \lambda_2 \geq\cdots
\ee
and any finite $N>0$, we here use the $N$ largest eigenvalues $\lambda_0,\ldots,\lambda_{N-1}$ to introduce
\bb
\label{eq:chi}
\widetilde{\chi}_N\qty(T)&\coloneqq \Tr T -\sum_{n=0}^{N-1}\lambda_n =\sum_{n=N}^{\infty}\lambda_n\, ,
\ee
and $\widetilde{\chi}_0\qty(T)\coloneqq \Tr T$ for $N=0$. For a $D$-dimensional positive semidefinite operator $M$ with $\lambda_0\geq\lambda_1\geq\cdots\geq\lambda_{D-1}\geq 0$, we can define $\widetilde{\chi}_N(M)$ by taking formally $\lambda_{D}=\lambda_{D+1}=\cdots=0$. With our function $\widetilde{\chi}_N$ at hand, we are now ready to generalize of the key technical result~\cite[Lemma~1]{Gheorghiu2008} to our infinite-dimensional setting. 

\begin{lemma}[{(Generalization of Ref.~\cite[Lemma~1]{Gheorghiu2008} to infinite-dimensional systems)}] \label{lemma:lemma_1}
Let $\mathcal{H}$ be any separable Hilbert space.
For any finite $N>0$, any pair of bounded operators $A,B\in\mathcal{B}\qty(\mathcal{H})$, and any positive semidefinite compact operator
\begin{equation}
K = \sum_{n=0}^{\infty} k_n \ketbra{n} \in\mathcal{K}_+\qty(\mathcal{H})\, , \qquad k_0 \geq k_1 \geq\cdots\, ,
\end{equation}
define
\begin{equation}
K_N \coloneqq \sum_{n=N}^{\infty} k_n \ketbra{n}\, .
\end{equation}
Then, it holds that
\begin{equation}
\widetilde{\chi}_N\qty[A K B B^\dag K A^\dag]\leq\Tr\qty[A K_N B B^\dag K_N A^\dag],
\end{equation}
where $\widetilde{\chi}_N$ is defined by~\eqref{eq:chi}.
\end{lemma}

In our proof of Lemma~\ref{lemma:lemma_1}, an infinite-dimensional version of the Schur--Horn Theorem~\cite{KAFTAL20103115} plays an essential role. In both the finite- and the infinite-dimensional setting, this latter result characterizes the relationship between the diagonal and the spectrum of positive semidefinite compact operators in terms of majorization. 

\begin{lemma}[{(Infinite-dimensional Schur--Horn Theorem~\cite[Proposition~6.4]{KAFTAL20103115})}]
\label{lem: inf_dim_SH}
Let $\mathcal{H}$ be any separable Hilbert space. 
For any positive integer $N > 0$, any positive semidefinite compact operator $M$, and any set of orthonormal vectors $\{\ket{v_n} \in \mathcal{H}: n = 0,1,\dots,N-1\}$, it holds that 
\begin{equation}
\sum_{n = 0}^{N-1} \lambda_n \geq \sum_{n = 0}^{N-1} \bra{v_n}M\ket{v_n}, 
\end{equation}
where $\lambda_0 \geq \lambda_1 \geq \cdots \geq \lambda_{N-1}$ are the $N$ largest eigenvalues of $M$. 
\end{lemma}

Using Lemma~\ref{lem: inf_dim_SH}, we can now prove Lemma~\ref{lemma:lemma_1}. 

\begin{proof}[Proof of Lemma~\ref{lemma:lemma_1}]
Let $\Pi_N$ be the projector onto the orthogonal complement of the range of $A\qty(K-K_N)$. We have that
\bb \label{eq:rank_bound}
&\rank \qty(\mathds{1}-\Pi_N)=\dim\qty(\mathrm{range}\qty(A\qty(K-K_N)))\leq N\, ,\qquad \Pi_N A K=\Pi_N A K_N\, .
\ee
Let ${\qty(\lambda_n)}_{n=0,1,\ldots}$ denote the sequence of eigenvalues of $A K B B^\dag K A^\dag$ sorted in the descending order. Then, it holds that
\bb
\widetilde{\chi}_N[A K B B^\dag K A^\dag]\ &= \Tr\qty[A K B B^\dag K A^\dag]-\sum_{n=0}^{N-1}\lambda_n \\
&= \Tr\qty[\Pi_N A K B B^\dag K A^\dag\Pi_N]-\left(\sum_{n=0}^{N-1}\lambda_n - \Tr\qty[\qty(\mathds{1}-\Pi_N) A K B B^\dag K A^\dag\qty(\mathds{1}-\Pi_N)]\right) \\
&\leq\Tr\qty[\Pi_N A K B B^\dag K A^\dag\Pi_N] \\
&= \Tr\qty[\Pi_N A K_N B B^\dag K_N A^\dag\Pi_N] \\
&\leq \Tr\qty[AK_N B B^\dag K_N A^\dag]\, , 
\ee
where the first inequality follows from~\eqref{eq:rank_bound} due to the infinite-dimensional version of the Schur--Horn Theorem (Lemma~\ref{lem: inf_dim_SH}), and the last follows from the operator inequality $\Pi_N\leq\mathds{1}$. To apply Lemma~\ref{lem: inf_dim_SH}, we observed that $A K BB^\dagger K A^\dagger$ is a positive semidefinite operator by definition, and it is also a compact operator because $K$ is compact and $A$ and $B$ are bounded~\cite{RUDIN}.
\end{proof}

Using the above lemma, we can generalize Ref.~\cite[Theorem~2]{Gheorghiu2008} so as to obtain the following result. 

\begin{lemma}[(Majorization theorem for countably separable operations)] \label{lemma:prop_1}
Let $\mathcal{H}_A$ and $\mathcal{H}_B$ be two separable Hilbert spaces. For any positive integer $N\in \N_+$, any bipartite pure state $\psi_{AB} = \ketbra{\psi}_{AB}$ with local reduction $\psi_A = \Tr_B \psi_{AB}$, and any family $\qty(L_k\otimes M_k)_{k=1,2,\ldots}$ of product bounded operators acting on $\mathcal{H}_A \otimes \mathcal{H}_B$ with
\bb
\label{R_operator}
R\coloneqq\sum_{k=1}^\infty L_k^\dag L_k^{\vphantom{\dag}}\otimes M_k^\dag M_k^{\vphantom{\dag}}\,
\ee
satisfying $\left\|R\right\|_\infty<\infty$, it holds that
\begin{equation}
\label{eq:majorization_bound}
\sum_{k=1}^\infty \widetilde{\chi}_N\qty(\Tr_B \!\left[L_k\otimes M_k\, \psi_{AB}\, L_k^\dag\otimes M_k^\dag\right]) \leq \widetilde{\chi}_N\qty(\psi_A) \left\|R\right\|_\infty,
\end{equation}
where $\widetilde{\chi}_N$ is defined as~\eqref{eq:chi}, and $\left\| \cdot \right\|_\infty$ is the operator norm defined by~\eqref{eq:operator_norm}.
\end{lemma}

\begin{proof}
Let $\ket{\psi}=\sum_{n=0}^{\infty}\sqrt{p\qty(n)} \ket{n}\otimes\ket{n}$ be the Schmidt decomposition of $\ket{\psi}$. Then, setting
\begin{equation}
D \coloneqq \sum_{n=0}^{\infty} \sqrt{p(n)} \ketbra{n},
\end{equation}
due to the map-state duality, we have that
\bb
\Tr_B \!\left[L_k\otimes M_k\, \psi_{AB}\, L_k^\dag\otimes M_k^\dag\right] &= \sum_{n,m=0}^\infty \sqrt{p(n)\,p(m)} \,\Tr_B \!\left[L_k\otimes M_k\, \ketbraa{nn}{mm}\, L_k^\dag\otimes M_k^\dag\right] \\
&= \sum_{n,m=0}^\infty \sqrt{p(n)\,p(m)}\, \bra{m}M_k^\dag M_k^{\vphantom{\dag}} \ket{n}\, L_k^{\vphantom{\dag}}\! \ketbraa{n}{m} L_k^\dag \\
&= \sum_{n,m=0}^\infty \sqrt{p(n)\,p(m)}\, \bra{n}M_k^\intercal M_k^{*} \ket{m}\, L_k^{\vphantom{\dag}}\! \ketbraa{n}{m} L_k^\dag \\
&= L_k^{\vphantom{\dag}} D M_k^\intercal M_k^* D L_k^\dag\, ,
\label{lemma:prop_1_proof_eq1}
\ee
where $M_k^\intercal$ and $M_k^{*}$ are the operators defined by the relations $\bra{n}M_k^\intercal\ket{m}=\bra{m}M_k\ket{n}$ and $\bra{n}M_k^{*}\ket{m}=\bra{n}M_k\ket{m}^{*}$ for all $n$ and $m$, respectively, and $\bra{n}M_k\ket{m}^{*}$ is the complex conjugate of the complex number $\bra{n}M_k\ket{m}$.
Therefore
\bb
\sum_{k=1}^\infty \widetilde{\chi}_N\qty(\Tr_B \!\left[L_k\otimes M_k\, \psi_{AB}\, L_k^\dag\otimes M_k^\dag\right])\ &\eqt{(i)}\ \sum_{k=1}^\infty \widetilde{\chi}_N\qty(L_k^{\vphantom{\dag}} D M_k^\intercal M_k^* D L_k^\dag) \\
&\leqt{(ii)}\ \sum_{k=1}^\infty \Tr\qty[L_k^{\vphantom{\dag}} D_N M_k^\intercal M_k^* D_N L_k^\dag]  \\
&\eqt{(iii)}\ \sum_{k=1}^\infty \bra{\psi_N} L_k^\dag L_k^{\vphantom{\dag}} \otimes M_k^\dag M_k^{\vphantom{\dag}} \ket{\psi_N} \\
&\eqt{(iv)}\ \bra{\psi_N} R \ket{\psi_N} \\
&\leqt{(v)}\ \braket{\psi_N} \|R\|_\infty \\ 
&\eqt{(vi)}\ \widetilde{\chi}_N\qty(\psi_A)\, \|R\|_\infty\, .
\ee
Here: (i)~follows from~\eqref{lemma:prop_1_proof_eq1}; (ii)~from an application of Lemma~\ref{lemma:lemma_1} with $A = L_k$, $B=M_k^\intercal$, and $K=D$, which is compact because it is the square root of a trace-class operator; (iii)~can be verified with a calculation almost identical to~\eqref{lemma:prop_1_proof_eq1} except for the substitution $\ket{\psi}\mapsto \ket{\psi_N} \coloneqq \sum_{n=N}^{\infty}\sqrt{p\qty(n)} \ket{n}\otimes\ket{n}$ and $D\mapsto D_N \coloneqq \sum_{n=N}^{\infty}\sqrt{p\qty(n)} \ketbra{n}$; in~(iv), we remembered the definition of $R$, reported in~\eqref{R_operator}; (v)~holds because of the definition of operator norm; and finally in~(vi), we noted that $\braket{\psi_N} = \sum_{n=N}^\infty p(n) = \widetilde{\chi}_N\qty(\sumno_{n=0}^\infty p(n) \ketbra{n}) = \widetilde{\chi}_N\qty(\psi_A)$. This concludes the proof. 
\end{proof}

Using the above technical lemmas, we can characterize the probabilistic pure state transformations under infinite-dimensional countably separable operations in terms of a majorization condition, thereby generalizing the finite-dimensional result of Ref.~\cite{Gheorghiu2008}.

\begin{prop}[(Majorization condition for probabilistic pure state transformation under countably separable channels)] \label{prop:theorem_1}
Let $\mathcal{H}_A$ and $\mathcal{H}_B$ be separable Hilbert spaces. Given any bipartite pure state $\ket{\psi}_{AB}$ on $AB$, if there exists a countably separable instrument that maps $\ket{\psi}_{AB}$ into an ensemble of pure states $\qty{p\qty(x),\ket{\phi_x}_{AB}}_{x=1,2,\ldots}$, then for all positive integers $N$ it holds that
\begin{equation}
\widetilde{\chi}_N\qty(\psi_A) \geq \sum_{x=1}^\infty p\qty(x)\,\widetilde{\chi}_N\big(\phi_x^A\big)\, ,
\end{equation}
where $\widetilde{\chi}_N$ is defined by~\eqref{eq:chi}.
\end{prop}

\begin{proof}
Let $\psi_{AB} = \ketbra{\psi}_{AB}$ and
$\phi_x^{AB} = \ketbra{\phi_x}_{AB}$.
By assumption, we can find a countably separable channel $\EE\in \sep_\N$ such that
\bb
\EE(\psi_{AB}) = \sum_x p(x)\, \phi_x^{AB} \otimes \ketbra{x}_X\, ,
\label{prop:theorem_1_proof_eq1}
\ee
where, by convention, we assign the classical system $X$ to Bob (the actual choice is immaterial). Countable separability means that we can find a Kraus decomposition $\EE (\cdot) = \sum_k L_k^{\vphantom{\dag}} \otimes M_k^{\vphantom{\dag}} (\cdot) L_k^\dag \otimes M_k^\dag$, where $L_k$ acts on $\mathcal{H}_A$, $M_k$ maps $\mathcal{H}_B$ into $\mathcal{H}_B\otimes \mathcal{H}_X$, and due to the trace-preserving condition
\bb
\sum_k L_k^\dag L_k^{\vphantom{\dag}} \otimes M_k^\dag M_k^{\vphantom{\dag}} = \id_{AB}.
\label{prop:theorem_1_proof_eq1.5}
\ee
For any fixed $x=1,2,\ldots$, consider the operators
$M_{k,x}\coloneqq 
\bra{x} M_k$
acting on $\mathcal{H}_B$ and defined by the relation $\bra{\xi}M_{k,x}\ket{\zeta}_B \coloneqq \big(\bra{\xi}\otimes\bra{x}\big) M_k \ket{\zeta}$, for all $\ket{\xi}_B,\ket{\zeta}_B\in \mathcal{H}_B$. Then
\bb
\sum_k L_k^{\vphantom{\dag}} \otimes M_{k,x}^{\vphantom{\dag}}\, \psi_{AB}\, L_k^\dag \otimes M_{k,x}^\dag = \qty(\id_{AB}\otimes\bra{x}_X) \EE\qty(\psi_{AB}) \qty(\id_{AB}\otimes\ket{x}_X) = p(x)\, \phi_x^{AB}\, ,
\label{prop:theorem_1_proof_eq2}
\ee
which implies that
\bb
L_k \otimes M_{k,x} \ket{\psi}_{AB} = \sqrt{p(x)}\, c(k|x)\, \ket{\phi_x}_{AB}\qquad \forall\ k=1,2,\ldots ,
\label{prop:theorem_1_proof_eq3}
\ee
where $c(k|x)\in \C$ is a sequence of complex numbers such that
\bb
\sum_k \left|c(k|x)\right|^2 = 1\, .
\label{prop:theorem_1_proof_eq4}
\ee
We can now write
\bb
\sum_x p\qty(x)\,\widetilde{\chi}_N\big(\phi_x^A\big)\ &\eqt{(i)}\ \sum_{x,k} p\qty(x) \left|c(k|x)\right|^2 \,\widetilde{\chi}_N\big(\phi_x^A\big) \\
&\eqt{(ii)}\ \sum_{x,k} \widetilde{\chi}_N\qty(\Tr_B \!\left[ L_k^{\vphantom{\dag}} \otimes M_{k,x}^{\vphantom{\dag}}\, \psi_{AB}\, L_k^\dag \otimes M_{k,x}^\dag \right]) \\
&\leqt{(iii)}\ \widetilde{\chi}_N \qty(\psi_A)\, .
\label{prop:theorem_1_proof_eq5}
\ee
Here, (i)~comes from~\eqref{prop:theorem_1_proof_eq4}, (ii)~from~\eqref{prop:theorem_1_proof_eq3}, (iii)~from Lemma~\ref{lemma:prop_1} with the observation, based on~\eqref{prop:theorem_1_proof_eq1.5}, that
\bb
\left\|\sumno_{x,k} L_k^\dag L_k^{\vphantom{\dag}}\otimes M_{k,x}^\dag M_{k,x}^{\vphantom{\dag}}\right\|_\infty &= \left\|\sumno_{k} L_k^\dag L_k^{\vphantom{\dag}}\otimes M_{k}^\dag \left(\sumno_x \ketbra{x}\right) M_{k}^{\vphantom{\dag}}\right\|_\infty\\
&= \left\|\sumno_{k} L_k^\dag L_k^{\vphantom{\dag}}\otimes M_{k}^\dag M_{k}^{\vphantom{\dag}}\right\|_\infty\\
&= \left\|\id_{AB}\right\|_\infty\\
&=1<\infty
\ee
This completes the proof. 
\end{proof}

To say something about the monotonicity of the entanglement of formation, we need to work with the entropy rather than with the function $\widetilde{\chi}_N$ for representing majorization. To translate the relations given by Proposition~\ref{prop:theorem_1} into statements involving entropy, we here develop a strategy that, to the best of our knowledge, has not been worked out in the literature so far. The linchpin of it is the following integral representation of the quantum entropy in finite- or infinite-dimensional systems, which we believe is of independent interest. It was inspired by recent work done in this direction --- but in finite dimension --- in Refs.~\cite{Frenkel2023,Jencova2023,Hirche2023}.

\begin{lemma}[(Integral representation of the quantum entropy in terms of the function for representing majorization)] \label{lemma:integral_representation_entropy}
Let $\rho\in \mathcal{D}(\mathcal{H})$ be a quantum state over a (possibly infinite-dimensional) separable Hilbert space $\mathcal{H}$. Then its quantum entropy $S(\rho) = -\Tr \rho \log_2\rho$ defined by~\eqref{entropy}, be it finite or infinite, admits the integral representation
\bb
S(\rho) = \frac{1}{\ln 2} \left(\int_0^1 \frac{\dd \mu}{\mu}\, \min_{N\in \N} \big\{ \widetilde{\chi}_N(\rho) + N\mu\big\} - 1\right) ,
\label{integral_representation_entropy}
\ee
where $N$ runs on the non-negative integers, and $\widetilde{\chi}_N$ is defined by~\eqref{eq:chi}.
\end{lemma}

Note that the function $\N \ni N \mapsto \widetilde{\chi}_N(\rho) + N\mu$ admits a minimum for all values of $\mu > 0$, simply because $\lim_{N\to\infty} \left\{ \widetilde{\chi}_N(\rho) + N\mu \right\} = +\infty$. The proof of Lemma~\ref{lemma:integral_representation_entropy} is relegated to Appendix~\ref{app:proof_integral_representation_entropy}. We use the above integral formula to deduce the following corollary.

\begin{cor}[(Conversion of conditions from majorization to quantum entropy)] \label{cor:from_chi_to_S}
Let $\mathcal{H},\mathcal{H}'$ be separable Hilbert spaces. Consider a state $\rho\in \mathcal{D}(\mathcal{H})$ and a countable ensemble $\left\{ p(x),\, \omega_x\right\}_{x=1,2,\ldots}$ of states $\omega_x\in \mathcal{D}(\mathcal{H}')$ on $\mathcal{H}'$. Here, $p$ is a probability distribution. Assume that
\bb
\widetilde{\chi}_N (\rho) \geq \sum_{x=1}^\infty p(x)\, \widetilde{\chi}_N(\omega_x)\qquad \forall\ N=1,2,\ldots\,,
\label{chi_inequalities}
\ee
where $\widetilde{\chi}_N$ is defined by~\eqref{eq:chi}. Then, it holds that
\bb
S(\rho) \geq \sum_{x=1}^\infty p(x)\, S(\omega_x)\, ,
\label{S_inequality}
\ee
where the quantum entropy $S$ is defined by~\eqref{entropy}.
\end{cor}
\begin{proof}
For each $\mu\in (0,1)$, let $N_\mu\in \N$ be such that $\min_{N\in \N} \left\{ \widetilde{\chi}_N(\rho) + N\mu\right\} = \widetilde{\chi}_{N_\mu}(\rho) + N_\mu\mu$. (As discussed above, such $N_\mu$ must exist.) Then 
\bb
\min_{N\in \N} \left\{ \widetilde{\chi}_N(\rho) + N\mu \right\} &= \widetilde{\chi}_{N_\mu}(\rho) + N_\mu \mu \\
&\geqt{(i)} \sum_x p(x)\, \widetilde{\chi}_{N_\mu}(\omega_x) + N_\mu \mu \\
&= \sum_x p(x) \left( \widetilde{\chi}_{N_\mu}(\omega_x) + N_\mu\mu\right) \\
&\geq \sum_x p(x) \min_{N\in \N} \left\{ \widetilde{\chi}_{N}(\omega_x) + N \mu\right\} ,
\label{from_chi_to_S_proof_eq1}
\ee
where in~(i) we used the assumption~\eqref{chi_inequalities}. For all $x$, the function $\mu \mapsto \min_{N\in \N} \left\{ \widetilde{\chi}_{N}(\omega_x) + N \mu\right\}$ is measurable due to Lemma~\ref{lemma:integral_representation_entropy}; furthermore, the series $\sum_x p(x) \min_{N\in \N} \left\{ \widetilde{\chi}_{N}(\omega_x) + N \mu\right\}$ converges for every $\mu$, as it is composed of positive terms and $\sum_x p(x) \min_{N\in \N} \left\{ \widetilde{\chi}_{N}(\omega_x) + N \mu\right\} \leq \sum_x p(x) \widetilde{\chi}_{0}(\omega_x) = 1$. This implies that the sum of the series is itself a measurable function of $\mu$. 

Integrating both sides of~\eqref{from_chi_to_S_proof_eq1} thus yields
\begin{align}
1 + (\ln2)\, S(\rho) &\eqt{(ii)} \int_0^1 \frac{\dd \mu}{\mu} \min_{N\in \N} \left\{ \widetilde{\chi}_N(\rho) + N\mu\right\} \nonumber \\
&\geq \int_0^1 \frac{\dd \mu}{\mu} \sum_x p(x) \min_{N\in \N} \big\{ \widetilde{\chi}_N(\omega_x) + N\mu\big\} \nonumber \\
&\eqt{(iii)}\, \sum_x p(x)\, \int_0^1 \frac{\dd \mu}{\mu} \min_{N\in \N} \big\{ \widetilde{\chi}_N(\omega_x) + N\mu\big\} \\
&\eqt{(iv)}\, \sum_x p(x) \big( 1+ (\ln2)\, S(\omega_x)\big) \nonumber \\
&=\, 1 + (\ln 2) \sum_x p(x)\, S(\omega_x)\, . \nonumber
\end{align}
Here, (ii)~and~(iv) follow from Lemma~\ref{lemma:integral_representation_entropy}, and in~(iii) we exchanged sum and integral using Tonelli's theorem, which is applicable because $p(x)\geq 0$ and $\min_{N\in \N} \big\{ \widetilde{\chi}_N(\omega_x) + N\mu\big\}\geq 0$. Simple algebraic manipulations lead to~\eqref{S_inequality}, thereby completing the proof.
\end{proof}

Combining the above result with Proposition~\ref{prop:theorem_1}, we obtain the following key statement. Note that if the target ensemble is composed of a single pure state only, then the deduction of the monotonicity of entanglement entropy from the majorization relation would be an immediate corollary of Ref.~\cite[Theorem~2.2]{LI2013384}. The main point, of course, is that we consider arbitrary pure-state ensembles as targets. 
We also note that infinite-dimensional bipartite pure state transformations under LOCC and stochastic LOCC (SLOCC) have been studied in Refs.~\cite{PhysRevA.70.050301,DBLP:journals/qic/OwariBNM08,AsaD:2017}, but the techniques found there are not directly applicable here since our focus is on a different family of operations, i.e., 
separable operations.

\begin{prop}[(Monotonicity of the entanglement of formation under 
separable operations)] \label{prop:monotonicity_entanglement_formation_under_sep}
Let $\mathcal{H}_A$ and $\mathcal{H}_B$ be any (possibly infinite-dimensional) separable Hilbert spaces. Then, the entanglement of formation defined on $\mathcal{D}\qty( \mathcal{H}_A\otimes \mathcal{H}_B)$ by~\eqref{eof} is monotonic under general separable channels (Definition~\ref{def:separable_channels}), i.e.,
\bb
E_f\big(\EE(\rho)\big) \leq E_f(\rho)\qquad \forall\ \rho = \rho_{AB} \in \mathcal{D}\qty(\mathcal{H}_A \otimes \mathcal{H}_B)\, , \qquad \forall\ \EE\in \sep\, .
\label{monotonicity_entanglement_formation_under_sep}
\ee
\end{prop}

\begin{proof}
We start by observing that given any bipartite pure state $\ket{\psi}_{AB}$ on $AB$, if there exists a countably separable instrument that transforms $\ket{\psi}_{AB}$ into an ensemble of pure states $\qty{p\qty(x),\ket{\phi_x}_{AB}}_{x=1,2,\ldots}$, then it holds that
\begin{equation}
S\qty(\psi_A) \geq \sum_{x=1}^\infty p\qty(x)S\big(\phi_x^A\big)\, .
\label{monotonicity_entropy_under_sep}
\end{equation}
Indeed, the above inequality follows directly by combining Proposition~\ref{prop:theorem_1} and Corollary~\ref{cor:from_chi_to_S}. 

The next step is to prove that~\eqref{monotonicity_entanglement_formation_under_sep} holds when the continuous version of the entanglement of formation $E_f$ (defined in~\eqref{eof}) is replaced by its discrete version $\widetilde{E}_f$ (defined in~\eqref{EoF_tilde}), and $\EE$ is assumed to be \emph{countably} separable (Definition~\ref{def:separable_channels}). To this end, one can employ a standard argument. Let $\qty(L_k\otimes M_k)_{k=1,2,\ldots}$ be a countable family of Kraus operators for $\EE\in \sep_\N$. Then, given an arbitrary pure-state decomposition of $\rho$ into pure states, say $\rho_{AB} = \sum_x p(x)\, \psi_x^{AB}$, we can define 
\bb
\ket{\phi_{x,k}}\coloneqq \frac{1}{\sqrt{q(k|x)}}\, \qty(L_k\otimes M_k) \ket{\psi_x}\, , \qquad q(k|x)\coloneqq \bra{\psi_x} L_k^\dag L_k^{\vphantom{\dag}} \otimes M_k^\dag M_k^{\vphantom{\dag}} \ket{\psi_x}\, .
\ee
By construction, there is a countably separable instrument that maps each $\ket{\psi_x}$ into the ensemble $\big\{q(k|x),\, \ket{\phi_{x,k}}\big\}_{k=1,2,\ldots}$, and hence by~\eqref{monotonicity_entropy_under_sep}, we have that
\bb
S\big(\psi_x^A\big) \geq \sum_{k=1}^\infty q(k|x)\, S\big(\phi_{x,k}^A\big)\, .
\ee
Multiplying by $p(x)$ and summing over $x$, we obtain
\bb
\sum_x p(x)\, S\big(\psi_x^A\big) \geq \sum_{x,k} p(x)\, q(k|x)\, S\big(\phi_{x,k}^A\big) \geq \widetilde{E}_f\qty(\mathcal{E}\qty(\rho))\, ,
\label{monotonicity_entanglement_formation_under_sep_proof_eq2}
\ee
where the last inequality follows because
\bb
\EE(\rho) = \sum_{x,k} p(x)\, q(k|x)\, \phi_{x,k}
\ee
is a valid discrete decomposition of $\EE(\rho)$ into pure states. Since~\eqref{monotonicity_entanglement_formation_under_sep_proof_eq2} holds for any discrete decomposition of $\rho$ into pure states, we can take the infimum over those, which yields $\widetilde{E}_f(\rho)$ on the leftmost side and thereby completes the proof of the monotonicity of $\widetilde{E}_f$ under countably separable channels.

We are now ready to prove~\eqref{monotonicity_entanglement_formation_under_sep} in full generality.\footnote{The argument that follows is courtesy of Maksim E.\ Shirokov (personal communication).} Let $\rho = \rho_{AB}$ be an arbitrary state. If $E_f(\rho) = +\infty$ there is nothing to prove; therefore, let us assume that $E_f(\rho)<\infty$. We now invoke Ref.~\cite[Lemma~3]{Shirokov2005} (with the substitutions $\Phi \mapsto \Tr_B$ and hence $\hat{H}_\Phi \mapsto E_f$, cf.\ definition of $\hat{H}_\Phi$ in Ref.~\cite[p.~9]{Shirokov2005}) to conclude the existence of a sequence $(\rho_n)_n$ of states $\rho_n = \rho_n^{AB}$ such that $S\qty(\rho_n^A) < \infty$ for all $n$, $\lim_{n\to\infty} \left\|\rho_n - \rho\right\|_1 = 0$, and furthermore $\lim_{n\to\infty} E_f(\rho_n) = E_f(\rho)$. In Ref.~\cite[Lemma~3]{Shirokov2005} it is further shown that the $\rho_n$ can be taken of finite rank, but we shall not make use of this particular fact here. Now, given a separable channel $\EE \in \sep$, consider a sequence $\big(\EE_n\big)_n$ of countably separable channels $\EE_n\in \sep_\N$ such that $\EE_n \tendsn{s} \EE$ with respect to the strong operator topology. Then
\bb
E_f\big(\EE(\rho)\big) &\leqt{(i)} \liminf_{n\to\infty} E_f\big(\EE_n(\rho_n)\big) \\
&\leqt{(ii)} \liminf_{n\to\infty} \widetilde{E}_f\big(\EE_n(\rho_n)\big) \\
&\leqt{(iii)} \liminf_{n\to\infty} \widetilde{E}_f(\rho_n) \\
&\eqt{(iv)} \liminf_{n\to\infty} E_f(\rho_n) \\
&\eqt{(v)} E_f(\rho)\, .
\label{monotonicity_entanglement_formation_under_sep_proof_eq4}
\ee
The justification of the derivation is as follows: (i)~is due to the lower semi-continuity of $E_f$, recalled in~\eqref{lsc_EoF}, which is applicable because
\bb
\left\|\EE_n(\rho_n) - \EE(\rho)\right\|_1 &\leq \left\|\EE_n(\rho_n) - \EE_n(\rho)\right\|_1 + \left\|\EE_n(\rho) - \EE(\rho) \right\|_1 \\
&\leq \left\|\rho_n - \rho\right\|_1 + \left\|\EE_n(\rho) - \EE(\rho) \right\|_1 \\
&\tendsn{} 0\, ;
\ee
here, we observed that the first term on the second line goes to zero due to the aforementioned~\cite[Lemma~3]{Shirokov2005}, and the second does the same because of strong convergence of $\EE_n$ to $\EE$; continuing, (ii)~follows by the elementary inequality~\eqref{EoF_smaller_EoF_tilde}; (iii)~is an application of the monotonicity of $\widetilde{E}_f$ under countably separable channels, established above; (iv)~is due to the equality in~\eqref{eq:min_S_A_B}, which holds because $\rho_n$ has finite local entropy for all $n$; and finally~(v) is again due to Ref.~\cite[Lemma~3]{Shirokov2005}. This concludes the justification of~\eqref{monotonicity_entanglement_formation_under_sep_proof_eq4} and, therefore, the proof.
\end{proof}

\begin{rem}
\label{rem:monotonicity}
While the above proof shows along the way the monotonicity of the discrete version $\widetilde{E}_f$ of the entanglement of formation in~\eqref{EoF_tilde} under countably separable operations, it remains open whether $\widetilde{E}_f$ is also monotonic under general separable operations. Our contribution here is to prove that the continuous version $E_f$ in~\eqref{eof} also has this latter property. 
\end{rem}

\subsection{Proof of the converse part}
\label{subsec:proof_of_converse}

We are finally ready to state and prove the converse part of our main result, Theorem~\ref{cost_thm}.

\begin{prop}[(Converse part of entanglement cost in infinite-dimensional systems)] \label{cost_prop_under_sep}
Let $A$ and $B$ be quantum systems represented by separable Hilbert spaces. For any bipartite state $\rho_{AB}\in\mathcal{D}\qty(\HH_A \otimes \HH_B)$ with finite local quantum entropy $\min\qty{S\qty(\rho_{A}),S\qty(\rho_{B})}<\infty$, it holds that
\bb
E_{c,\,\sep}\qty(\rho)\geq E_f^\infty\qty(\rho)\, ,
\ee
where $\sep$ denotes the set of separable operations (Definition~\ref{def:separable_channels}), and $E_f^\infty$ is the regularized entanglement of formation given by~\eqref{regularized_eof}.
\end{prop}

\begin{proof}
Without loss of generality, we are going to assume that $S(\rho_A) <\infty$. (The same argument also works when $S(\rho_B) <\infty$, up to exchanging the roles of $A$ and $B$.) By Lemma~\ref{associated_Hamiltonian_lemma}, we can construct a grounded Hamiltonian $H_A$ on $A$ that satisfies the Gibbs hypothesis and such that $E \coloneqq \Tr \qty[\rho_{A} H_A] < \infty$. Let us now extend this to a Hamiltonian on $\mathcal{H}_A^{\otimes n} = \mathcal{H}_{A^n}$ by setting
\bb
H_{A^n} \coloneqq H_A\otimes \id \otimes \ldots \id\ +\ \id\otimes H_A \otimes \id\otimes \ldots \otimes \id\ +\ \ldots\ +\ \id \otimes \ldots\otimes \id\otimes H_A\, .
\label{n_copy_Hamiltonian}
\ee
It is straightforward to see that $H_{A^n}$, just like $H_A$, is a grounded Hamiltonian that satisfies the Gibbs hypothesis.

With the definition of entanglement cost in~\eqref{eq:entanglement_cost}, fix any $\delta>0$. Let $r$ be an achievable rate for entanglement dilution of $\rho_{AB}$ under separable operations, satisfying
\bb
\label{eq:r_E_c}
r = E_{c,\,\sep}\qty(\rho) + \delta\, .
\ee
Then, by definition, for any $\epsilon>0$ there exists a sufficiently large integer $n_0$ and a sequence of separable protocols $\EE_n \in \sep$ 
with the property that for all integers $n>n_0$
\bb
\frac12 \left\|\rho_n - \rho^{\otimes n}\right\|_1 \leq \e\, , \qquad \rho_n \coloneqq \EE_n\qty(\Phi^{\otimes \lfloor rn\rfloor}).
\ee
Following Lemma~\ref{alsc_formation_lemma}, we set $\epsilon^\prime \coloneqq \sqrt{\e \qty(2-\e)}$.

With all the above techniques in hand, the crux of the proof reduces to the following chain of inequalities, joining into a similar route to the finite-dimensional case of Ref.~\cite{Hayden-EC}: for all $n>n_0$,
\bb
\floor*{rn}\ &\eqt{(i)}\ E_f\Big(\Phi^{\otimes \floor*{rn}}\Big) \\
&\geqt{(ii)}\ E_f 
\big(\rho_n\big)\\
&\geqt{(iii)}\ E_f\big(\rho^{\otimes n}\big) - \epsilon^\prime F_{H_{A^n}}\left( \frac{n E}{\epsilon^\prime} \right) - g(\epsilon^\prime) \\
&\eqt{(iv)}\ E_f\big(\rho^{\otimes n}\big) - n \epsilon^\prime F_{H_A}\left( \frac{E}{\epsilon^\prime} \right) - g(\epsilon^\prime)\, .
\label{cost_proof_crux}
\ee
Here, (i)~is simply the normalization of the entanglement of formation for (finite-dimensional) maximally entangled states; (ii)~follows from the monotonicity under the separable channel $\EE_n$, i.e., Proposition~\ref{prop:monotonicity_entanglement_formation_under_sep}; in~(iii), we used Lemma~\ref{alsc_formation_lemma} with respect to the Hamiltonian $H_{A^n}$, observing that $\Tr\qty[ \rho^{\otimes n} H_{A^n}] = n \Tr\qty[ \rho_{A} H_A] = nE$; finally in~(iv), we noticed that since by~\eqref{n_copy_Hamiltonian} $H_{A^n}$ models a non-interacting system, it holds that
\bb
\gamma_{H_{A^n},\,\beta} = \gamma_{H_A,\,\beta}^{\otimes n}\, ,
\ee
and hence $\beta_{H_{A^n}}(nE') = \beta_{H_A}(E')$ for all $E'\geq 0$, implying that
\bb
F_{H_{A^n}}(n E') &= S\big(\gamma_{H_{A^n},\, \beta_{H_{A^n}}(nE')}\big)\\
&= S\big(\gamma_{H_{A^n},\, \beta_{H_A}(E')}\big) = S\big(\gamma_{H_{A},\, \beta_{H_{A}}(E')}^{\otimes n}\big)\\
&= n\, S\big(\gamma_{H_{A},\, \beta_{H_{A}}(E')}\big)\\
&= n\, F_{H_A}(E')\, .
\ee
This completes the justification of~\eqref{cost_proof_crux}. Dividing both sides of the relation~\eqref{cost_proof_crux} by $n$ and taking the limit $n\to\infty$, we have that
\bb
r = \lim_{n\to\infty} \frac{\floor*{rn}}{n} \geq E_f^\infty(\rho) - \epsilon^\prime F_{H_A}\left( \frac{E}{\epsilon^\prime} \right) .
\ee
We are now ready to take the limit $\e\to 0^+$, i.e., $\e^\prime\to 0^+$, in the above inequality. Invoking Lemma~\ref{F_H_o_of_E_lemma} yields immediately
\bb
r \geq  E_f^\infty(\rho) - \lim_{\e^\prime\to 0^+} \epsilon^\prime F_{H_A}\left( \frac{E}{\epsilon^\prime} \right) = E_f^\infty(\rho)
\ee
for every fixed $E\geq 0$. Therefore, due to~\eqref{eq:r_E_c}, we have that
\begin{equation}
E_{c,\,\sep}\qty(\rho)
\geq E_f^\infty\qty(\rho) - \delta\, .
\end{equation}
Since this holds for an arbitrarily small choice of $\delta>0$, it follows that
\begin{equation}
E_{c,\,\sep}\qty(\rho) \geq E_f^\infty\qty(\rho)\, ,
\end{equation}
which concludes the proof.
\end{proof}

\begin{rem}
For readers who would prefer to work with the discrete version of the entanglement of formation $\widetilde{E}_f$ instead of its continuous counterpart $E_f$, it is possible to prove the above statement also by using $\widetilde{E}_f$. Indeed, the above proof works in the same way by replacing $E_{c,\,\sep}$ with $E_{c,\,\sep_\N}$, as these two quantities are anyway equal due to Lemma~\ref{same_cost_lemma}. The semi-continuity bound in Lemma~\ref{alsc_formation_lemma} holds also for $\widetilde{E}_f$~\cite[Proposition~4.B]{Shirokov2022}. Moreover, the advantage of working with $\widetilde{E}_f$ and $\sep_\N$ may be that the monotonicity statement in Proposition~\ref{prop:monotonicity_entanglement_formation_under_sep} becomes considerably simpler to establish, as an inspection of the proof of Proposition~\ref{prop:monotonicity_entanglement_formation_under_sep} reveals. 
But, as discussed in Remark~\ref{rem:monotonicity}, our contribution here is to establish the desired monotonicity of $E_f$ under separable rather than countably separable operations, and use this property to construct our proof.
\end{rem}

\section{Conclusion} \label{sec:conclusion}

Infinite-dimensional systems, whether natural or engineered, are common in physics. We commonly manipulate entanglement in such infinite-dimensional physical systems in experiments.
However, most of the significant results on asymptotic entanglement transformation established in the history of entanglement theory are simply lost as soon as one transitions to infinite-dimensional systems, due to the unavailability of key mathematical techniques that the finite-dimensional analyses rely on. 

Here we have solved the fundamental problem of characterizing an operational measure of entanglement, the entanglement cost~\cite{Bennett-distillation, Bennett-distillation-mixed, Bennett-error-correction,Hayden-EC, Buscemi2011}, in the general infinite-dimensional setting. To wit, we have proved that the entanglement cost of any infinite-dimensional state $\rho_{AB}$ with finite quantum entropy on either $A$ or $B$ 
is given by the regularized entanglement of formation. This extends the foundational finite-dimensional result of Refs.~\cite{Hayden-EC, Buscemi2011} to the case of infinite-dimensional quantum systems. We believe that the restriction to states with finite local entropy is physically justified, as states with finite local energy always obey it, for any choice of Hamiltonian that satisfies the Gibbs hypothesis.

The techniques used in the finite-dimensional analysis in Ref.~\cite{Hayden-EC}, such as strong typicality, monotonicity, and asymptotic continuity, are no longer directly applicable if we allow manipulation of infinite-dimensional systems. Nevertheless, we constructed an optimal protocol for entanglement dilution of infinite-dimensional mixed states by employing a known notion of typicality for countably infinite alphabets and a new ``wasteful'' protocol that deals specifically with the ``rare'' part of the state. 
Remarkably, the protocol is implementable with one-way LOCC and a finite amount of classical communication assisted by a finite number of ebits, yet it still approximately outputs copies of the infinite-dimensional target state.

We also proved the converse bound, which implies that this protocol is optimal even among all protocols that employ infinite-dimensional separable operations. To establish the converse, we have introduced techniques for proving the monotonicity of the entanglement of formation under countably separable operations, generalizing the results on finite-dimensional LOCC and SEP presented in Refs.~\cite{PhysRevLett.83.1455, PhysRevLett.84.4781, Vidal2000, Gheorghiu2008}. Among these techniques there is a new integral representation for the quantum entropy in infinite-dimensional systems, which was inspired by recent works on integral representations of quantum $f$-divergences. The other key piece of techniques we have used is the recent semi-continuity bound of the entanglement of formation for infinite-dimensional states in Ref.~\cite{Shirokov2022}, which is used in place of the conventional asymptotic continuity, appropriate only for finite-dimensional states~\cite{Nielsen2000, tightuniform, Mark2020}.

A fundamental open question that we leave for future investigation is whether other operational measures of more general quantum resources on infinite-dimensional systems~\cite{Kuroiwa2020, Kuroiwa2021, taming-PRL, taming-PRA} can also enjoy a characterization in terms of entropic quantities, in a similar way to what we did here. Also, from a more practical perspective, it would be interesting to consider infinite-dimensional multipartite entanglement manipulation on
quantum networks~\cite{PhysRevA.96.032330, YamasakiPhD, spee2021transformations}.

\section*{Acknowledgments}
L.L.\ thanks Maksim E.\ Shirokov and Mark M.\ Wilde for useful comments on the first version of this manuscript, and in particular the former for providing an insightful argument to complete the last part of the proof of Proposition~\ref{prop:monotonicity_entanglement_formation_under_sep}. H.Y.\ acknowledges JST PRESTO Grant Number JPMJPR201A, JPMJPR23FC, JSPS KAKENHI Grant Number JP23K19970, and MEXT Quantum Leap Flagship Program (MEXT QLEAP) JPMXS0118069605, JPMXS0120351339\@. K.K.\ was supported by a Mike and Ophelia Lazaridis Fellowship, a Funai Overseas Scholarship, and a Perimeter Residency Doctoral Award. P.H.\ acknowledges support from ARO (award W911NF2120214), DOE (Q-NEXT), CIFAR and the Simons Foundation. LL acknowledges financial support from the European Union under the European Research Council (ERC Grant Agreement No.~101165230) and from MIUR (Ministero dell'Istruzione, dell'Universit\`a e della Ricerca) through the project `Dipartimenti di Eccellenza 2023--2027' of the `Classe di Scienze' department at the Scuola Normale Superiore. 

Part of this work was carried out during the workshop ``Quantum resources: from mathematical foundations to operational characterisation'' held in Singapore in December~2022. Other key discussions happened at the conference ``It from Qubit'', held at the Perimeter Institute for Theoretical Physics in July--August~2023.

\bigskip
\noindent \textbf{Data availability statement.} No data sets were generated or analyzed during this study.

\medskip
\noindent \textbf{Competing interests.} The authors declare no competing interests.

\appendix

\section{Proof of lemma on Hamiltonian associated with a state}
\label{sec:proof_associated_hammltonian_lemma}

In this appendix, we present a simplified and fully self-contained proof of the very handy Lemma~\ref{associated_Hamiltonian_lemma}, derived by us before we became aware of Ref.~\cite[Proposition~4]{Shirokov-1}, which features a slightly more general statement.

\begin{lemma} \label{series_lemma}
Let $\qty(a_n)_{n\in\N}$ be a sequence of non-negative numbers $a_n\geq 0$ such that $\sum_{n=0}^\infty a_n < \infty$. Pick some $c > 1$. Then there exists another sequence $\qty(b_n)_{n\in \N}$ such that $b_n\geq 1$ for all $n$, $\lim_{n\to\infty} b_n = +\infty$ but
\bb
\sum_{n=0}^\infty a_n b_n \leq c \sum_{n=0}^\infty a_n < \infty\, .
\ee
\end{lemma}

\begin{proof}
It suffices to prove the claim for a fixed $c > 1$. Indeed, suppose that this has been done and that a sequence $\qty(b_n)_{n\in \N}$ with the required properties has been constructed. Given some other $c' < c$, writing $c' = p + \qty(1-p) c$ for some $p\in \qty(0,1)$, it suffices to construct the new sequence $b'_n \coloneqq p + \qty(1-p) b_n$, which satisfies the claim for $c'$ instead of $c$.

We will, therefore, fix $c=5$. Due to homogeneity we can also restrict to sequences $\qty(a_n)_{n\to\infty}$ satisfying $\sum_{n=0}^\infty a_n=1$. Finally, without loss of generality, we can assume that $\qty(a_n)_{n\in \N}$ does not become identically $0$ eventually in $n$ --- otherwise, the claim trivializes. Now, set
\bb
b_n \coloneqq 1 - \log_2 \left( \sumno_{m=n}^\infty a_m \right) .
\ee
Clearly, $b_n\geq 1$ and $\lim_{n\to\infty} b_n=+\infty$. Moreover, defining the sequence of non-decreasing numbers $\qty(n_k)_{k\in \N}$ such that 
\bb
n_k \coloneqq \min\left\{ n\in \N:\ \sumno_{m=n}^\infty a_m \leq 2^{-k} \right\} ,
\ee
so that $n_0=0$, we see that
\bb
\sum_{n=0}^\infty a_n \qty(b_n-1) &= \sum_{k=0}^\infty \sum_{n=n_k}^{n_{k+1}-1} a_n \left( - \log_2 \left( \sumno_{m=n}^\infty a_m \right)\right) \\
&\leq \sum_{k=0}^\infty \sum_{n=n_k}^{n_{k+1}-1} a_n \left( - \log_2 \left( \sumno_{m=n_{k+1}-1}^\infty a_m \right)\right) \\
&\leq \sum_{k=0}^\infty \qty(k+1) \sum_{n=n_k}^{n_{k+1}-1} a_n \\
&\leq \sum_{k=0}^\infty \qty(k+1) \sum_{n=n_k}^{\infty} a_n \\
&\leq \sum_{k=0}^\infty \qty(k+1)\, 2^{-k} \\
&= 4\, ,
\ee
so that indeed $\sumno_{n=0}^\infty a_n b_n \leq 5$, completing the proof.
\end{proof}

We are now ready to provide a full proof of Lemma~\ref{associated_Hamiltonian_lemma}.

\begin{proof}[Proof of Lemma~\ref{associated_Hamiltonian_lemma}]
It is well known that~(a) implies~(b), simply because the Gibbs state $\gamma_\beta = e^{-\beta H}/\Tr [e^{-\beta H}]$ maximizes the entropy for a given energy constraint, and hence $S(\rho) \leq S(\gamma_\beta)$, where $\beta$ is such that $\Tr \gamma_\beta H = \Tr \rho H$. (A solution to this equation must exist, because in infinite-dimensional systems $\lim_{\beta\to 0^+} \Tr \gamma_\beta H = +\infty$.) Whatever the value of $\beta$ obtained in this way, by expanding the logarithm it is immediate to check that $S(\gamma_\beta) < \infty$.

The converse statement (that (b)~implies~(a)) is less clear. To prove it, write the spectral decomposition of $\rho$ as $\rho = \sum_{n=0}^\infty p\qty(n) \ketbra{n}$, where $\qty{\ket{n}}_{n\in \N}$ is an orthonormal basis of $\mathcal{H}$. Since $(\ln 2)\, S\qty(\rho) = \sum_{n=0}^\infty p\qty(n) \ln \frac{1}{p\qty(n)}<\infty$, using Lemma~\ref{series_lemma} we can construct a sequence $\qty(b_n)_{n\in \N}$ such that $b_n\geq 1$ for all $n$, $\lim_{n\to\infty} b_n = +\infty$, and $\sum_{n=0}^\infty b_n p\qty(n) \ln \frac{1}{p\qty(n)} < \infty$. We now define the operator
\bb
H\coloneqq \sum_{n=1}^\infty b_n \ln \frac{1}{p\qty(n)} \ketbra{n}
\ee
with domain 
\bb
\dom\qty(H) \coloneqq \left\{ \ket{\psi} = \sumno_{n=0}^\infty \psi_n \ket{n}\in \mathcal{H}:\ \ \sumno_{n=1}^\infty b_n^2 \left(\ln \frac{1}{p\qty(n)}\right)^2 |\psi_n|^2 < \infty \right\} .
\ee
It is elementary to verify that $H$ is indeed self-adjoint, and that $\min\spec\qty(H) = \bra{0}H\ket{0} = 0$. Furthermore, for every $\beta>0$ it holds that
\bb
\Tr [e^{-\beta H}] &= 1 + \sum_{n=1}^\infty p\qty(n)^{\beta b_n} \\
&= 1 + \sum_{n\geq 1,\ \beta b_n \leq 1} p\qty(n)^{\beta b_n} + \sum_{n\geq 1,\ \beta b_n > 1} p\qty(n)^{\beta b_n} \\
&\leq 1 + \left|\left\{n\geq 1:\ \beta b_n \leq 1\right\}\right| + \sum_{n=1}^\infty p\qty(n) \\
&= 2 + \left|\left\{n\geq 1:\ \beta b_n \leq 1\right\}\right| < \infty\, ,
\ee
where the finiteness of $\left|\left\{n\geq 1:\ \beta b_n \leq 1\right\}\right|$ follows from the fact that $\lim_{n\to\infty} \beta b_n = +\infty$. This shows that $H$ satisfies the Gibbs hypothesis. Finally, it holds by construction that
\bb
\Tr [\rho H] = \sum_{n=1}^\infty b_n p\qty(n) \ln \frac{1}{p\qty(n)} < \infty\, ,
\ee
completing the proof.
\end{proof}

\section{Proof of the integral representation of the entropy}

\label{app:proof_integral_representation_entropy}

We give a proof of Lemma~\ref{lemma:integral_representation_entropy} on the integral representation of the entropy.

\begin{manuallemma}{\ref{lemma:integral_representation_entropy}}
Let $\rho\in \mathcal{D}(\mathcal{H})$ be a quantum state over a (possibly infinite-dimensional) separable Hilbert space $\mathcal{H}$. Then, its quantum entropy $S(\rho) = -\Tr \rho \log_2\rho$ defined by~\eqref{entropy}, be it finite or infinite, admits the integral representation
\bb
S(\rho) = \frac{1}{\ln 2} \left(\int_0^1 \frac{\dd \mu}{\mu}\, \min_{N\in \N} \big\{ \widetilde{\chi}_N(\rho) + N\mu\big\} - 1\right) ,
\label{integral_representation_entropy_appendix}
\ee
where $N$ runs on the non-negative integers, and $\widetilde{\chi}_N$ is defined by~\eqref{eq:chi}.
\end{manuallemma}

\begin{proof}
Without loss of generality,
we assume that the eigenvalues of $\rho$ are sorted in non-increasing order, i.e., 
\bb
p_0\geq p_1\geq p_2\geq \ldots. 
\label{integral_representation_entropy_proof_eq1}
\ee
Note that $\sum_n p_n = 1$, hence $p_n \tendsn{} 0$. In fact, it even holds that 
\bb
n\,p_n \tendsn{} 0\, .
\label{integral_representation_entropy_proof_eq2}
\ee
We learned this simple fact from Ref.~\cite{na_n_goes_to_0}, and we repeat here the proof found there for completeness. By the Cauchy condensation test, since $\sum_{n=0}^{\infty} p_n$ converges the same must be true for $\sum_{k=0}^{\infty}2^k\, p_{2^k}$, and therefore $2^k p_{2^k} \tends{}{k\to\infty} 0$. Now, for any $n\in \N_+$ we define $k_n \coloneqq \floor*{\log_2 n}$. Since $2^{k_n} \leq n \leq 2^{k_n+1}$, by monotonicity~\eqref{integral_representation_entropy_proof_eq1} we deduce that
\bb
p_{2^{k_n}} \leq p_n \leq p_{2^{k_n+1}}\, .
\ee
Multiplying by $n$, taking the limit $n\to\infty$, and using the fact that $k_n \tendsn{} \infty$ completes the proof of~\eqref{integral_representation_entropy_proof_eq2}.

We now go back to the proof of the integral formula. We start by observing that the function $\mu \mapsto \min_{N\in \N} \big\{ \widetilde{\chi}_N(\rho) + N\mu\big\}$, being the pointwise minimum of continuous function, is upper semi-continuous and in particular Lebesgue measurable. Hence, the integral on the right-hand side of~\eqref{integral_representation_entropy_appendix} is well defined.

Remember Abel's formula for summation by parts, which states that
\bb
\sum_{n=0}^N a_n b_n = a_N B_N + \sum_{n=0}^{N-1} \big(a_n - a_{n+1}\big) B_n
\label{integral_representation_entropy_proof_eq3}
\ee
for any two sets of numbers $a_1,\ldots, a_N$ and $b_1,\ldots, b_N$, where $B_m \coloneqq \sum_{k=0}^m b_k$. Setting $a_n = p_n$ and $b_n=1$ yields
\bb
\sum_{n=0}^N p_n = (N+1)\, p_N + \sum_{n=0}^{N-1} (n+1) \big(p_n - p_{n+1}\big)\, .
\ee
Since $(N+1)\, p_N\tends{}{N\to\infty} 0$ by~\eqref{integral_representation_entropy_proof_eq2}, we deduce that
\bb
\sum_{n=0}^{\infty} (n+1) \big(p_n - p_{n+1}\big) = \sum_{n=0}^\infty p_n = 1\, .
\label{integral_representation_entropy_proof_eq4}
\ee
By definition of $\widetilde{\chi}_n$ in~\eqref{eq:chi}, setting $a_n = \log_2 \frac{1}{p_n}$ and $b_n = p_n$ gives instead
\bb
\sum_{n=0}^N p_n \log_2 \frac{1}{p_n} &= \left(\log_2 \frac{1}{p_N}\right) \sum_{n=0}^N p_n + \sum_{n=0}^{N-1} \left(\log_2 \frac{p_{n+1}}{p_n} \right) \sum_{k=0}^n p_k \\
&= \left(\log_2 \frac{1}{p_N}\right) \sum_{n=0}^N p_n + \sum_{n=0}^{N-1} \left(\log_2 \frac{p_{n+1}}{p_n} \right) \left( 1 - \widetilde{\chi}_{n+1}(\rho)\right) \\
&= \left(\log_2 \frac{1}{p_N}\right) \sum_{n=0}^N p_n + \log_2 \frac{p_N}{p_0} + \sum_{n=0}^{N-1} \left(\log_2 \frac{p_{n}}{p_{n+1}} \right) \widetilde{\chi}_{n+1}(\rho) \\
&= -\left(\log_2 \frac{1}{p_N}\right) \widetilde{\chi}_{N+1}(\rho) + \log_2\frac{1}{p_0} + \sum_{n=0}^{N-1} \left(\log_2 \frac{p_{n}}{p_{n+1}} \right) \widetilde{\chi}_{n+1}(\rho)\, .
\label{integral_representation_entropy_proof_eq5}
\ee
To take the limit $N\to\infty$ in~\eqref{integral_representation_entropy_proof_eq5}, consider first the case where $S(\rho)<\infty$. Then, it holds that
\bb
0\leq \left(\log_2 \frac{1}{p_N}\right) \widetilde{\chi}_{N+1}(\rho) = \sum_{n=N+1}^\infty p_n \log_2 \frac{1}{p_N} \leq \sum_{n=N+1}^\infty p_n \log_2 \frac{1}{p_n} \tends{}{N\to\infty} 0\, ,
\ee
implying via~\eqref{integral_representation_entropy_proof_eq5} that
$\log_2\frac{1}{p_0}+\sum_{n=0}^{\infty} \left(\log_2 \frac{p_{n}}{p_{n+1}} \right) \widetilde{\chi}_{n+1}(\rho) = S(\rho)$.
If $S(\rho) = +\infty$ instead, then again due to~\eqref{integral_representation_entropy_proof_eq5}
\bb
\sum_{n=0}^{N-1} \left(\log_2 \frac{p_{n}}{p_{n+1}} \right) \widetilde{\chi}_{n+1}(\rho) = \sum_{n=0}^N p_n \log_2 \frac{1}{p_n} + \left(\log_2 \frac{1}{p_N}\right) \widetilde{\chi}_{N+1}(\rho) \geq \sum_{n=0}^N p_n \log_2 \frac{1}{p_n} \tends{}{N\to\infty} \infty\, .
\ee
This proves that the identity
\bb
S(\rho) = \log_2\frac{1}{p_0} + \sum_{n=0}^{\infty} \left(\log_2 \frac{p_{n}}{p_{n+1}} \right) \widetilde{\chi}_{n+1}(\rho)
\label{integral_representation_entropy_proof_eq6}
\ee
holds for all quantum states $\rho$, both of finite and of infinite quantum entropy. 

\begin{note}
If only finitely many eigenvalues of $\rho$ are nonzero, so that, say, $p_d = p_{d+1} = \ldots = 0$, then the correct way to interpret the above equation is by setting $\left(\log_2 \frac{p_{n}}{p_{n+1}} \right) \widetilde{\chi}_{n+1}(\rho) = 0$ for all $n\geq d-1$, which is only natural since $\widetilde{\chi}_{n+1}(\rho) = 0$ for these values of $n$. This is related to the fact that we set the value of the function $-x \log_2 x$ to be $0$ for $x=0$.
\end{note}

We can now write
\bb
\int_0^1 \frac{\dd \mu}{\mu}\, \min_{N\in \N} \big\{ \widetilde{\chi}_N(\rho) + N\mu\big\} \ &\eqt{(i)}\ \int_{p_0}^{1} \frac{\dd \mu}{\mu}\, \min_{N\in \N} \big\{ \widetilde{\chi}_N(\rho) + N\mu\big\} + \sum_{n=0}^\infty \int_{p_{n+1}}^{p_n} \frac{\dd \mu}{\mu}\, \min_{N\in \N} \big\{ \widetilde{\chi}_N(\rho) + N\mu\big\} \\
&\eqt{(ii)}\ \ln\frac{1}{p_0} + \sum_{n=0}^\infty \int_{p_{n+1}}^{p_n} \frac{\dd \mu}{\mu}\, \big( \widetilde{\chi}_{n+1}(\rho) + (n+1)\mu\big) \\
&=\ \ln\frac{1}{p_0} + \sum_{n=0}^\infty \left( (n+1) \left(p_n - p_{n+1}\right) + \widetilde{\chi}_{n+1}(\rho) \ln\frac{p_n}{p_{n+1}}\right) \\
&\eqt{(iii)}\ 1 + (\ln 2)\, S(\rho)\, .
\ee
We now justify the above derivation: (i)~is proved by noting that since the integrand is non-negative, we have that
\bb
\int_0^{p_0} \frac{\dd \mu}{\mu}\, \min_{N\in \N} \big\{ \widetilde{\chi}_N(\rho) + N\mu\big\} &= \lim_{N\to\infty} \int_{p_N}^{p_0} \frac{\dd \mu}{\mu}\, \min_{N\in \N} \big\{ \widetilde{\chi}_N(\rho) + N\mu\big\} \\
&= \lim_{N\to\infty} \sum_{n=0}^{N-1} \int_{p_{n+1}}^{p_n} \frac{\dd \mu}{\mu}\, \min_{N\in \N} \big\{ \widetilde{\chi}_N(\rho) + N\mu\big\}\, ;
\ee
in~(ii), we noted that for all $n\in \N$ and $p_{n+1}\leq \mu\leq p_n$, the difference
\bb
\widetilde{\chi}_{N+1}(\rho) + (N+1)\mu - \left(\widetilde{\chi}_N(\rho) + N\mu\right) = \mu - p_N
\ee
is non-positive for $N\leq n$, and non-negative for $N\geq n+1$; consequently,
\bb
\label{eq:inf_n_explicit}
\min_{N\in \N} \big\{ \widetilde{\chi}_N(\rho) + N\mu\big\} = \widetilde{\chi}_{n+1}(\rho) + (n+1)\mu\, ;
\ee
moreover, for the same reasons, it is clear that for $\mu\geq p_0$ we have that $\min_{N\in \N} \big\{ \widetilde{\chi}_N(\rho) + N\mu\big\} = \widetilde{\chi}_0(\rho) = 1$; finally, in~(iii) we used~\eqref{integral_representation_entropy_proof_eq4} and~\eqref{integral_representation_entropy_proof_eq6}.
This completes the proof.
\end{proof}

\bibliographystyle{unsrt}
\bibliography{biblio}

\end{document}